\documentclass[12pt, table]{amsart}

\usepackage{amsmath,amscd,amsfonts,amssymb, xcolor, bbm, bm,
braket, xcolor, calc, enumerate} 
\theoremstyle{plain}
\numberwithin{equation}{section} 
\newtheorem{theorem}{Theorem}[section]   
\newtheorem{prop}[theorem]{Proposition}
  
\newtheorem{rem}[theorem]{Remark}
\newtheorem{ass}[theorem]{Assumption}
\newtheorem{lemma}[theorem]{Lemma}
\newtheorem{cor}[theorem]{Corollary}
\theoremstyle{definition}
\newtheorem{eg}[theorem]{Example}

\newcommand{\tr}{\mathrm{Tr\,}}
\newcommand{\dd}{\mathrm{d}}

\newcommand{\ketbra}[1]{\ensuremath{\left|#1\right\rangle
\hspace{-3pt}\left\langle #1\right|}}

\newcommand{\vep}{\varepsilon}

\newcommand{\wt}{\widetilde}
\newcommand{\wh}{\widehat}
\newcommand{\var}{\mathrm{Var}}

\begin{document}
\title{The likelihood operator and Fisher information  in quantum probability}
\author{Kalyan B. Sinha}
\address{Jawaharlal Nehru Centre for Advanced Scientific Research, Jakkur, Bangalore
560 064, India} 
\address{Indian Statistical Institute, Stat-Math. Unit, R V College Post, Bengaluru 560 059, India}
\email[Kalyan B. Sinha]{kbs@jncasr.ac.in}

\author{Ritabrata Sengupta}
\address{Department of Mathematical Sciences, Indian Institute of
Science Education \& Research (IISER) Berhampur, Transit Campus, Govt. ITI, 
Berhampur 760 010, Ganjam, Odisha, India}
\email[Ritabrata Sengupta]{rb@iiserbpr.ac.in}

\begin{abstract}
We study the problem of Quantum Likelihood 
Operators (LO)
and their connection with quantum  Fisher information (QFI).  It is observed that the present approaches to
this problem tacitly assume commutativity of parametrised density
matrix $\rho_\theta$ and its derivative, which, in general, need not be
true, and this has nontrivial consequences in QFI. As examples, we discuss the parametrised two-level system exhaustively, and, as a further example, the one-mode coherent states of an infinite-dimensional system.

\end{abstract}

\maketitle
\section{Introduction}\label{s1}
\par In quantum (non-commutative) probability theory, the notion of
states (described by a density matrix (or operator) $\rho$ in a
complex separable Hilbert space $\mathcal{H}$) and observables $X$
(often corresponds to bounded self-adjoint operators in
$\mathcal{H}$) play a fundamental role, replacing that of a probability
measure and of real-valued random variables in a suitable measure
space. If real $\lambda$ is an eigenvalue of $X$ with eigenprojection
$P$, then one say that a measurement of $X$  in the state $\rho$ would
find the value $\lambda$ with probability $\tr(\rho P)$. 

\par Often, the exact state $\rho$ of the system is not known, and as in
the case of the estimator of the parameters of a probability
distribution, one needs to consider a theory of estimation of the
\emph{ most likely} state among a given family, based on a set of observations on the
system. One such model, inspired by the classical theory is to assume
 the set of possible states is given by a family $\{\rho_\theta
\in \mathcal{B}_{1,+}(\mathcal{H})$ non-negative trace-class operators
in $\mathcal{H}$ with $\tr\rho_\theta =1|~\theta \in$ an open
interval $I \subseteq \mathbb{R}\}$ and the task is to estimate  the
value $\wh{\theta}$ of $\theta$ such that $\rho(\wh{\theta})$ is an
unbiased state. Some of the earliest studies in this direction are due
to Holevo \cite{holevo1}, Helstorm \cite{hels-67, hels-69, hels-73, hell-book}; here the authors write down versions of the
\emph{quantum likelihood operator} (QLO or just LO for short) leading
to the \emph{quantum Fisher information} (QFI) and associated quantum
Cramer-Rao bound (CR for short). In some of these references, as
well as in Parthasarathy \cite{krp-hyp}, a cost function associated with the
parameter space is defined, and its minimum is sought or the maximal
likelihood method is introduced for testing of \emph{quantum hypothesis}. One
should also mention the work of Krishna and Parthasarathy \cite{krishna-krp} and of
Parthasarathy \cite{krp21}  in which the ideas were extended to the case of
generalised measurements in finite-level quantum systems. 

\par The present work is presumed to be the first in a series to
study possible definitions of QLO, QFI, and their properties and
establish quantum CR bound in each case, with many of the results
valid for $\mathcal{H}$ of arbitrary dimensions.  

\par In Section \ref{s2}, we briefly revisit the standard procedure in
classical theory (in which the assumption that there exists a
$\theta$-independent measure $\mu(\cdot)$ on the measure space such that
every $\theta$-dependent probability measure $P_\theta(\cdot)$ is
absolutely continuous with respect to $\mu(\cdot)$, plays a major
role) and since the property of \emph{absolute continuity} among states  is not 
a simple concept in the quantum domain, the classical set-up is reformulated
obviating the necessity of having such a reference measure
$\mu(\cdot)$. This throws light on the path to adopt in the quantum
theory. In this process, two striking observations appear: firstly, a
two-dimensional example where though $\theta \mapsto \rho_\theta$ is
continuously differentiable with respect to the parameter $\theta \in I$, a bounded closed interval in $\mathbb{R}$, its constituent eigen-projections are not
even continuous, and secondly, for a differentiable state density
matrix $\rho_\theta$, its derivative $\rho_\theta'$ commutes with
$\rho_\theta$ for each $\theta$ if and only if the term containing
the derivatives of the eigen-projections vanish. In Section \ref{s3}, the
minimal set of assumptions necessary to begin a theory of parametric quantum
states are stated, the counter-example in two dimensions is
constructed, which forces a stronger hypothesis on the
$\theta$-dependence of $\rho_\theta$. The Boltzmann-von Neumann Logarithmic Derivative model is introduced first (in Section \ref{s3}), and its alternative three possible scenarios for the logarithmic derivative operator (including that of Helstrom) is studied in Section \ref{hels}. It is found that the logarithmic derivative operators, and consequently the QFIs, are all equal in all these models, if the eigenprojections are independent of the parameter $\theta$ (see Theorem \ref{th3.1} and Corollary \ref{cor4.4}). 
This section follows an explicit computation of the Cramér-Rao bound and a calculation of quantum Fisher Information (QFI) for all four of these models.

\par Section \ref{s5} explicitly addresses the scenario for infinite-dimensional systems. It extends the four models previously described to the case where $\dim \mathcal{H} =\infty$. In this article, it was noted repeatedly that the LD operators (in infinite-dimensional systems) are self-adjoint and unbounded, necessitating certain adjustments in the definitions of expectation in quantum theory to ensure that the resulting operators are of trace-class. For a physically relevant example of an infinite-dimensional scenario, we consider the one-parameter family of coherent states, and demonstrate explicit calculations of the QFI in two different models.
Section \ref{s6} concludes the paper with explicit calculations of QFI for two-level systems in all four models.  

\section{Revisiting classical theory}\label{s2}
\par We revisit the classical case first, then continue to the quantum
information theory.  Let $\{P_\theta: \theta \in \mathbb{R}\}$ be a
probability measure on sample space $\Omega$, which we take as a
compact topological space for simplicity. Let $\mu$ be a standard
Borel measure on $\Omega$ such that for all $\theta$, $P_\theta(\cdot)$ is absolutely
continuous with respect to $\mu(\cdot)$ with  the Radon–Nikodym
derivative $f_\theta(\omega), \, \forall \theta,\,
\omega \in \Omega$.

\par The {\bf likelihood function/estimator} is given by the natural logarithm of this function:
$\ell(\theta,\omega) \equiv \ln f_\theta(\omega)$. 
In general, some additional regularity conditions are assumed, viz 
the probability density function $f$ is in $C^2$ with respect to parameter $\theta$.
\par It is to be noted that, unlike the classical probability measures, the property of absolute continuity among quantum states is more complicated and 
 in this section, we want to construct a
framework avoiding any reference measure/state.

\par The expectation
of the derivative of the likelihood function, i.e.  the Logarithmic Derivative (LD) is given as:
$\mathbb{E}_\theta  \ell'(\theta, \omega) = \int P_\theta (\dd \omega)
\dfrac{\partial_\theta f_\theta(\omega)}{f_\theta(\omega)},$ 
where for convenience, we have set 
 \newline$\ell'(\theta, \omega)
\equiv \partial_\theta\ell(\theta, \omega)$.
This is not defined for $\omega$ for which  $f_\theta(\omega) =0$, though we note that for 
\newline $\Delta =
  \{\omega: f_\theta(\omega) =0\}$, 
$P_\theta(\Delta) = \int_\Delta f_\theta(\omega)  \mu(\dd \omega) =0,$
and the above expression should not create a problem. In fact,  the
expectation takes the form,
\[\mathbb{E}_\theta ~ \ell'(\theta, \omega) 
= \int \partial_\theta f_\theta(\omega) \mu(\dd \omega) = \int \partial_\theta  f_\theta(\omega) \mu(\dd \omega) = \partial_\theta \langle \mu(\cdot), f_\theta(\cdot)\rangle = \partial_\theta P(\Omega) =0, \] 
where differentiation is in the weak sense with respect to the pairing $\left \langle \mu(\cdot), f_\theta(\cdot) \right\rangle$. 

\par Next the \emph{classical Fisher Information}  is given as:
\begin{eqnarray}
I(\theta) = \var_\theta(\ell'(\theta, \cdot)) &=& \mathbb{E}_\theta [ \ell'(\theta, \cdot) - \mathbb{E}_\theta(\ell'(\theta, \cdot))]^2 \\
&=& \int P_\theta(\dd \omega)  \left( \frac{\partial_\theta f_\theta(\omega)}{f_\theta (\omega)} \right)^2 
= \int \mu(\dd \omega) \frac{(\partial_\theta f_\theta(\omega))^2}{f_\theta (\omega)}.\nonumber
\end{eqnarray}

\par For the maximality of the $\ell(\theta, \cdot)$, the second derivative is needed as well, the expectation of which is given by 
\begin{eqnarray*}
\mathbb{E}_\theta ~\partial_\theta^2 \ell(\theta, \omega) &=& \int
P_\theta (\dd \omega) \left[ \frac{\partial_\theta^2
f_\theta(\omega)}{f_\theta(\omega)} - \left( \frac{\partial_\theta
f_\theta(\omega)}{f_\theta(\omega)}\right)^2 \right]\\
&=& \int \mu(\dd \omega) \partial_\theta^2 f_\theta(\omega)  - \int
\frac{(\partial_\theta f_\theta(\omega))^2}{f_\theta(\omega)} \mu(\dd
\omega)\equiv \partial_\theta^2(P_\theta(\Omega)) -I(\theta) =  -I(\theta)
\le 0.
\end{eqnarray*} 
Thus the likelihood function $\ell(\theta,\omega)$ attains its local
maximum, in expectation.


 \par It may be noted that there is an implicit assumption in the
discussion above, viz.  $f_\theta(\omega) =0$ implies that
$\partial_\theta f_\theta(\omega) =0$. In other words, it is implicit
that not only $P_\theta(\cdot)$ is absolutely continuous with respect to $\mu$ (to ensure existence of $f_\theta(\cdot)$) but $\partial_\theta P_\theta(\cdot)$ needs to be absolutely continuous with respect to $P_\theta(\cdot)$ itself (to ensure meaning to $\frac{\partial_\theta f_\theta}{f_\theta}(\cdot)$). 

 \par Let $P_\theta(\cdot)$ be  a probability measure over $C_\mathbb{R}(\Omega)$
and let  $g \in C_{\mathbb{R}}(\Omega)$. Let
 \begin{equation}\label{eq3.1}
 \theta \longmapsto  \langle P_\theta, g \rangle = \langle P_\theta(\cdot), g(\cdot)\rangle \equiv \int
P_\theta(\dd \omega) g(\omega),
 \end{equation}
 be a $C^1$ for all  $g$. Then, by the Uniform Boundedness Principle and the Riesz' representation theorem \cite[pp 55, pp 235]{MR1009162}, there exists a
set function $\partial_\theta P_\theta(\cdot)$ of absolutely bounded
variation such that,
 \[ \left| \langle \partial_\theta P_\theta(\cdot), g(\cdot) \rangle \right| = C\|g\|_\infty, \quad \text{where $C$ is a constant}.\]

\par This implies that if one needs to build a classical theory without needing a reference measure, it is sufficient to assume that 
(i) the family $P_\theta(\cdot)$ is weak-$*$ continuously twice differentiable as above, and 
(ii) the derivative $\partial_\theta P_\theta(\cdot)$, being of absolute bounded variation on $\Omega$, is absolutely continuous with respect to $P_\theta(\cdot)$ itself.
Thus in such a case, there exists a unique real-valued measurable function $h(\theta, \cdot) \in L^1(\Omega, P_\theta)$ such that $\partial_\theta P_\theta (\Delta) =  \int_\Delta h(\theta, \omega) P_\theta (\dd \omega)$. This function $h(\theta, \cdot)$ plays the role of the logarithmic derivative (LD) mentioned earlier (in the traditional cases with a $\theta$-independent reference measure $\mu$). 
The likelihood function from this point of view is given by the
indefinite integral:
\[\ell (\theta, \omega) =\int^\theta h(\tau, \omega) \dd\tau. \]
In this scenario, the expectation  and variance of $h$  are given respectively by:
\begin{eqnarray*}
\mathbb{E}_\theta(\ell'(\theta,\cdot)) &=& \mathbb{E}_\theta(h(\theta,\cdot))= \int_\Omega P_\theta(\dd \omega) h(\theta, \omega)
= \partial_\theta \langle P_\theta (\cdot), 1 \rangle 
= \partial_\theta(P_\theta(\Omega)) =0, \\
 \var_\theta(h(\theta, \cdot))
&=& \mathbb{E}_\theta \left(h(\theta, \cdot)\right)^2
= \int P_\theta(\dd \omega) \left[ \frac{\partial_\theta P_\theta(\dd \omega)}{P_\theta(\dd \omega)}\right]^2
= \int P_\theta(\dd \omega) (h(\theta, \omega))^2
\equiv I(\theta) \ge 0,
\end{eqnarray*}
which is the (new form of) \emph{Fisher information}.  Finally, since by our assumption $h(\theta, \omega)$ is differentiable with respect to $\theta$,
\[\mathbb{E}_\theta (\ell''(\theta, \cdot)) = \int P_\theta (\dd \omega) ~\frac{\partial}{\partial \theta} h(\theta, \omega).\]
Furthermore,
\[
\partial_\theta \langle P_\theta(\cdot), h(\theta, \cdot) \rangle = \langle \partial_\theta P_\theta(\cdot), h(\theta, \cdot) \rangle + \langle P_\theta(\cdot), \partial_\theta h(\theta, \cdot)\rangle 
= \int P_\theta(\dd \omega) h(\theta, \omega)^2 + \mathbb{E}_\theta(\ell''(\theta,\cdot)).
\]
On the other hand, for the left-hand side above
\[ \partial_\theta \langle P_\theta(\cdot), h(\theta, \cdot) \rangle = \partial_\theta \left( \int P_\theta(\dd \omega) \frac{\partial_\theta P_\theta(\cdot)}{P_\theta(\cdot)}(\omega)\right) = \partial_\theta^2P_\theta(\Omega) =0,\]
showing that $\mathbb{E}_\theta(\ell''(\theta, \cdot)) = -I(\theta) \le 0$ i.e. the  maximality of the likelihood  function in expectation.

\begin{rem}
There is an implicit
hypothesis that the null set does not depend on $\theta$. In fact, the
assumption of the existence of a reference measure $\mu$ independent
of $\theta$ in the earlier discussion precisely ensures that.
\end{rem}

\section{Quantum case}\label{s3}
\par Taking the cue from the second part of Section \ref{s2}, we find that
the main mathematical issues for defining an analogous theory for a
quantum system are: 
\begin{enumerate}[(i)]
\item differentiability of the state $\rho_\theta$ with respect to the parameter $\theta$, which is understood as  weak differentiability with respect to $\theta$; this will imply that $\rho_\theta'$ is trace-class self adjoint, 
\item the possible representation (realisation) of \emph{the
logarithmic derivative of $\rho_\theta$}, the minimum requirement for
which is that the state density matrix $\rho_\theta$ has full rank,
i.e. $\mathrm{ker}(\rho_\theta) = \{0\},\, \forall \theta \in I$, an interval in $\mathbb{R}$.
\end{enumerate}
Thus $\rho_\theta = \sum_k \lambda_k(\theta) P_k(\theta)$, where the eigenvalues $\lambda_k(\theta) >0$ for each $k$ and for each $\theta \in I$, $\sum_k \lambda_k(\theta) =1$ for each $\theta$,  and $P_k(\theta)$ are the eigen-projections. 

\par If the density matrix $\rho_\theta$ is in a Hilbert space $\mathcal{H}$ with $\dim \mathcal{H} \equiv d<\infty$, then differentiability of $\lambda_k(\cdot)$ and $P_k(\cdot)$ clearly implies the differentiability of $\rho_\theta$. However, the converse is not true, as is shown in the following counterexample for $d=2$. 
\begin{eg}\label{exk}
Let $\mathcal{H} = \mathbb{C}^2$, and $\theta$ be the parameter such
that $\theta \in [-1,1] \subseteq \mathbb{R}$. Consider the density
function $\rho_\theta$ defined as
\begin{eqnarray}\label{exeq}
\rho_\theta &\equiv& \frac{1}{2} \left(1 + \exp\left( -\frac{1}{\theta^2}\right)\right)\begin{bmatrix}
\cos^2\left(\frac{1}{\theta}\right) & \cos\left(\frac{1}{\theta}\right) \sin\left(\frac{1}{\theta}\right) \nonumber\\
\cos\left(\frac{1}{\theta}\right) \sin\left(\frac{1}{\theta}\right) & \sin^2\left(\frac{1}{\theta}\right) \end{bmatrix} \\
&& + \frac{1}{2}  \left(1 - \exp\left( -\frac{1}{\theta^2}\right)\right)\begin{bmatrix}
\sin^2\left(\frac{1}{\theta}\right) & -\cos\left(\frac{1}{\theta}\right) \sin\left(\frac{1}{\theta}\right) \nonumber\\
-\cos\left(\frac{1}{\theta}\right) \sin\left(\frac{1}{\theta}\right) & \cos^2\left(\frac{1}{\theta}\right) \end{bmatrix}\\
&\equiv&  \lambda_1(\theta) P_1(\theta) + \lambda_2(\theta) P_2(\theta), \quad \text{ for $\theta \neq 0$, and }\\
\rho_0 &\equiv& \frac{1}{2} \begin{bmatrix} 1 & 0 \\ 0 & 0 \end{bmatrix} + \frac{1}{2} \begin{bmatrix} 0 & 0 \\ 0 & 1 \end{bmatrix} \quad \text{ for $\theta =0$.} \nonumber
\end{eqnarray}
Here $\lambda_j(\theta)$'s and $P_j(\theta)$'s, for $j =1,\,2$ are the eigenvalues and corresponding eigen-projections, given by 
\begin{alignat*}{2}
  \lambda_1(\theta)  &= \frac{1}{2} \left(1 + \exp\left( -\frac{1}{\theta^2}\right)\right),  \qquad       &   P_1(\theta)&= \begin{bmatrix}
\cos^2\left(\frac{1}{\theta}\right) & \cos\left(\frac{1}{\theta}\right) \sin\left(\frac{1}{\theta}\right) \\
\cos\left(\frac{1}{\theta}\right) \sin\left(\frac{1}{\theta}\right) & \sin^2\left(\frac{1}{\theta}\right) \end{bmatrix}, \\
   \lambda_2(\theta)  &= \frac{1}{2} \left(1 - \exp\left( -\frac{1}{\theta^2}\right)\right) ,        & P_2(\theta) &= \begin{bmatrix}
\sin^2\left(\frac{1}{\theta}\right) & -\cos\left(\frac{1}{\theta}\right) \sin\left(\frac{1}{\theta}\right) \\
-\cos\left(\frac{1}{\theta}\right) \sin\left(\frac{1}{\theta}\right) & \cos^2\left(\frac{1}{\theta}\right) \end{bmatrix}.
\end{alignat*}

It is easy to see that both $\lambda_1(\theta)$ and
$\lambda_2(\theta)$ are greater than or equal to $ 0$ for all
$\theta$. So $\rho_\theta$ is
faithful,\footnote{Note that a state $\phi: \mathcal{A} \to
\mathbb{C}$ is said to be faithful if $\phi(x^*x) =0$ implies that
$x=0$ for $x \in \mathcal{A}$.} i.e. $\mathrm{ker}(\rho_\theta)
=\{0\}$ for all $\theta \in [-1,1]$. Furthermore, $P_1(\theta)$ and
$P_2(\theta)$ are mutually orthogonal projections for all $\theta$, and $P_1(\theta)+P_2(\theta)=I_2$, i.e. the identity matrix. Write

\begin{eqnarray*}
\rho_\theta &=& \frac{1}{2} (P_1(\theta)+P_2(\theta)) + \frac{1}{2}\exp\left( -\frac{1}{\theta^2}\right) (P_1(\theta)- P_2(\theta))\\
&=& \frac{1}{2} I_2 + \frac{1}{2}\exp\left( -\frac{1}{\theta^2}\right)(P_1(\theta)- P_2(\theta)). 
\end{eqnarray*}
The the map $\rho_\theta \mapsto \rho_0 = \frac{1}{2} I_2$, as $|\theta| \to 0$ since the trigonometric functions above are bounded and $\exp\left( -\frac{1}{\theta^2}\right) \to 0$. In fact $[-1,1] \ni \theta \mapsto \rho_\theta$ is a $C^\infty$ function.
\par Observe that the eigenvalues as map $\theta \mapsto \lambda_1(\theta)$ and $\lambda_2(\theta)$ are both $C^\infty$ functions since
\[\frac{\dd^n}{\dd\theta^n} \exp\left( -\frac{1}{\theta^2}\right) \to 0 \quad \text{ as }\theta \to 0,\]
 so that both $\lambda_1(\theta)$ and $\lambda_2(\theta) \to \lambda_1(0)$ and $\lambda_2(0)$ respectively, each of  which equals $\frac{1}{2}$. 
\par On the other hand, the map to the eigen-projections, i.e. $\theta
\mapsto P_1(\theta)$ or $P_2(\theta)$ is not even continuous at
$\theta =0$, as the functions in $P_j(\theta)$ for $j=1, \,2$ does not
converge anywhere at $\theta =0$.
\end{eg}
\par Starting with the seminal paper of  Helstrom \cite{hels-73} many of the authors implicitly (or explicitly) assumed that in the canonical representation of $\rho_\theta$, only the eigenvalues depend on the parameter $\theta$,
not the eigen-projections.  This
counterexample \ref{exk} in $\mathbb{C}^2$ shows that merely the assumption of
differentiability $\theta \longmapsto \rho_\theta$ is not sufficient to
ensure the continuity of the eigen-projections of $\rho_\theta$ with
respect to $\theta$, and hence, it forces us to make a stronger
assumption for Hilbert spaces of all dimensions. 

\begin{ass}\label{as3.4}
Let $\{\rho_\theta : \theta \in I\}$ be a family of states in a complex separable Hilbert space $\mathcal{H}$.  We assume that (i) $\ker(\rho_\theta) =\{0\}$ for all $\theta \in $ interval $I \subseteq \mathbb{R}$ and (ii) $I \ni \theta \longmapsto \rho_\theta \in \mathcal{B}_{1\, +}$
 be a real analytic function. 
\end{ass}
\par This assumption ensures that in the absence of exceptional points \cite[p 64 -- 68 and 385 -- 387]{kato}, the eigenvalue and eigen-projections are analytic functions in $\theta$.\footnote{More generally, if there is an exceptional point, there may be several branches, and one needs to restrict to any of these branches for the same calculations.} Furthermore, since $\rho_\theta$ are bounded self-adjoint family, it also follows from the above that there exist unitary families $U(\theta)$ (for $\theta \in I \subset \mathbb{R}$) such that $P_k(\theta) =U(\theta) P_k(0) U(\theta)^*$ \cite[p 99 and p 386]{kato}. It can be seen easily from this that $\tr P_k(\theta) = \tr P_k(0) \equiv m_k \in \mathbb{Z}_+$, independent of $\theta$. Furthermore, under the hypothesis, one can differentiate inside the summation to get 
\begin{equation}\label{3.2}
\rho_\theta' \equiv \frac{\dd \rho_\theta}{\dd\theta} = \sum_{k =1}^ \infty \lambda_k'(\theta) P_k(\theta) + \sum_{k =1}^ \infty \lambda_k(\theta) P_k'(\theta).
\end{equation}

\par Next, a few general results related to the eigen-projections of the
family of states $\{\rho_\theta\}$ are given. These lead to a necessary and
sufficient condition when $\rho_\theta$ and  $\rho_\theta'$ commute
for each $\theta$.
\begin{lemma}\label{l1}
Let the family $\{\rho_\theta\}$ of states satisfy the Assumption 3.2.  Then
\begin{enumerate}[I.]
\item
\begin{enumerate}[(i)]
\item $P_j'(\theta) P_j(\theta) = P_j(\theta)^\perp P_j'(\theta)$,   $P_j'(\theta) P_j(\theta)^\perp =  P_j(\theta) P_j'(\theta)$, and \newline $P_j(\theta)' P_k(\theta) = - P_j(\theta) P_k(\theta)'$  $(j \neq k)$, 
\item $P_k(\theta)P_j'(\theta)P_k(\theta) =0$ for all $k$ and all $j$. Thus $P_j'(\theta)$ in the basis of $\{P_k(\theta)\}$ is an 
\newline off-diagonal operator, 
\item $(P_j'(\theta)P_k'(\theta)) P_j(\theta) = P_j(\theta)(P_j'(\theta) P_k'(\theta))^*$ for all $j,\,k$; $(P_j'(\theta))^2$ commutes with $P_j(\theta)$.
\end{enumerate}
If furthermore $\{P_j(\theta)\}$ is $C^2$, then 
\begin{enumerate}[(iv)]
\item $ P_j(\theta) P_k(\theta)'' + P_j(\theta)'' P_k(\theta) + 2 P_j(\theta)' P_k(\theta)' = \delta_{kj} P_j(\theta)''$. 
\end{enumerate}
\item $\rho_\theta$ commutes with $\rho_\theta'$ for all $\theta \in I$ if and only if $\sum_j \lambda_j P_j'(\theta) =0$. 
\end{enumerate}
\end{lemma}
\begin{proof}
As per discussion after Assumption 3.2,  $\lambda_k(\theta)$ and $P_k(\theta)$ are both $C^\infty$ in $\theta$. 
Differentiating the equation $(P_j(\theta))^2 =P_j(\theta)$ one gets 
\begin{eqnarray*}
P_j'(\theta) &=& P_j(\theta) P_j'(\theta) +P_j'(\theta) P_j(\theta) \\
\Rightarrow P_j'(\theta) P_j(\theta)^\perp &=& P_j(\theta) P_j'(\theta)
\end{eqnarray*}
Thus $P_j(\theta)P_j'(\theta)P_j(\theta) =0$ for all $j$. 
\par For $k \neq j, \, P_k(\theta)P_j(\theta) =0$. Differentiating we get
\[P_k'(\theta)P_j(\theta)+P_k(\theta) P_j'(\theta)=0.\]
Multiplying the right hand side by $P_k(\theta)$ we get $P_k(\theta)P_j'(\theta)P_k(\theta) =0, \quad \forall k \neq j$, which proves (ii). 
\par For $j \neq k$,  $P_j'(\theta)P_k'(\theta) P_j(\theta) = P_j'(\theta)(-P_k(\theta) P_j'(\theta)) = -(-P_j(\theta)P_k'(\theta)) P_j'(\theta) = P_j(\theta)P_k'(\theta) P_j'(\theta) = P_j(\theta)(P_j'(\theta) P_k'(\theta))^*. $ This proves the first part of (iii), and the second part follows by putting  $k=j$,  since $P_j'(\theta)^2$ is self-adjoint. 
\par The part (iv) is a consequence of differentiating twice the relation $P_k(\theta) P_j(\theta) = \delta_{kj} P_j(\theta)$. 
\par Proof of II:  $\rho_\theta'= \sum_{k=1}^\infty \lambda_k'(\theta) P_k(\theta) + \sum_{k=1}^\infty \lambda_k(\theta) P_k'(\theta) \equiv \mu + \nu$, (i.e. $\mu$'s are terms with $P_k(\theta)$'s while the $\nu$'s are terms with $P_k'(\theta)$'s). Thus, since $\mu$ commutes with $\rho_\theta$, $\rho_\theta'$ can commute with $\rho_\theta$ if and only if $\nu$ commutes with $\rho_\theta$ or equivalently (by spectral theorem) if and only if $\nu$ commutes with each $\{P_j(\theta)\}$. This implies that 
\[P_j(\theta) \nu = P_j(\theta) \left( \sum_{k=1}^\infty \lambda_k(\theta) P_k'(\theta)\right) = \nu P_j(\theta) =\left( \sum_{k=1}^\infty \lambda_k(\theta) P_k'(\theta)\right)P_j(\theta).\]
Since as a consequence of Assumption 3.2,  $\sum_k \lambda_k(\theta) P_k'(\theta)$ is absolutely, uniformly convergent in $\mathcal{B}_1$-norm, it follows that 
\[ P_j(\theta) \nu = P_j(\theta) \left( \sum_{k=1}^\infty \lambda_k(\theta) P_k'(\theta)\right) P_j(\theta) = \sum_{k=1}^\infty \lambda_k(\theta) \left ( P_j(\theta)P_k'(\theta) P_j(\theta)\right) =0,\]
by Lemma \ref{l1}(ii) for all $j$. That $\nu=0$ follows by summing over $j$ and noting that
\newline  $\nu = \left(\sum_{j=1}^\infty P_j(\theta) \right) \nu = \sum_{j=1}^\infty \left(P_j(\theta) \nu \right) =0$.

\end{proof}

\begin{theorem}\label{new}
For a one-parameter family of states $\rho_\theta$ with spectral
decomposition  be 
\newline $\rho_\theta =\sum_{k=1}^\infty \lambda_k(\theta) \sum_{r=1}^{m_k}
e_{k,r}(\theta) e_{k,r}(\theta)^*= \sum_{k=1}^\infty \lambda_k(\theta)
P_k(\theta)$, where $\lambda_k(\theta)$'s are eigenvalues, 
\newline $\{ e_{k,r}(\theta): 1 \le r \le m_k, 1 \le k < \infty\}$ are corresponding eigenvectors, and $P_k(\theta)$'s are corresponding eigen-projections of rank $m_k$.
If
$\rho_\theta' \in \mathcal{B}_1$, then $\sum_{k=1}^\infty |\lambda_k'(\theta)| <
\infty$. 
Conversely, if $\rho_\theta$ is real analytic in $\theta$ and if $\sum_{k=1}^\infty  |\lambda_k(\theta)|\sum_{r=1}^{m_r}  \| e_{k,r}'(\theta)\|$ converges uniformly in $\theta$, then  $\rho_\theta' \in \mathcal{B}_1$ and
\begin{equation}
\rho_\theta' = \sum_{k=1}^\infty \lambda_k'(\theta) P_k(\theta) + \sum_{k=1}^\infty \lambda_k(\theta) P_k'(\theta).
\end{equation}
\end{theorem}
\begin{proof}
\par Since $e_{k,r}(\theta)$ is an eigenvector, $(\rho_\theta -
\lambda(\theta))e_{k,r}(\theta)=0$. Differentiating we get 
\[ (\rho_\theta' - \lambda_k'(\theta))e_{k,r}(\theta) + (\rho_\theta - \lambda(\theta))e_{k,r}'(\theta)=0.\]
Taking innerproduct with $e_{k,r}(\theta)$, we get $\langle e_{k,r}(\theta), \rho_\theta'~ e_{k,r}(\theta)\rangle = \lambda_k'(\theta)$. This implies that 
\[\sum_{k=1}^\infty m_k |\lambda_k'(\theta)| \le \sum_{k=1}^\infty \sum_{r=1}^{m_k} |\langle e_{k,r}(\theta), \rho_\theta'~ e_{k,r}(\theta)\rangle| \le \|\rho_\theta'\|_1 < \infty, \quad \text{by assumption.}\] 

\par As observed above, 
$\|P_k(0)\|_1 = \tr P_k(0) = m_k$ and hence $\left\|  \sum_{k=1}^\infty \lambda_k'(\theta) P_k(\theta)\right \|_1 \le \sum_{k=1}^\infty |\lambda_k'(\theta)| m_k$, which converges uniformly by hypothesis. Similarly, the second term in the sum 
\[ \left\| \sum_{k=1}^\infty \lambda_k(\theta) P_k'(\theta) \right \| \le \sum_{k=1}^\infty \lambda_k(\theta) \|P_k'(\theta)\|_1 \le 2\sum_{k=1}^\infty  |\lambda_k(\theta)| \sum_{r=1}^{m_k}  \| e_{k,r}'(\theta)\|,\]
which also converges uniformly in $\theta$ by the hypothesis,  justifying the interchange of differentiation and infinite sum. 
\end{proof}

\par In this article, the various expressions for Fisher Information in quantum theory corresponding to different \emph{forms} of the \emph{Logarithmic Derivative} (LD) of the parametrised state $\{\rho_\theta: \theta \in I\}$ in a separable complex Hilbert space of arbitrary dimensions and their properties are studied.  Furthermore, for simplicity on presentation, the explicit dependence of $\theta$ will often be dropped, e.g. shall write $\rho$ for $\rho_\theta$, $\rho'$ for $\frac{\dd}{\dd\theta} \rho_\theta$, and $P_k'$ for $\frac{\dd}{\dd\theta} P_k(\theta)$, etc.

\par Logarithm and logarithmic derivatives are well-defined in cases of commutative functions and are routinely used,  in particular in classical estimation theory, as long as the domain does not contain zero. This is not so in quantum theory, even when the underlying Hilbert space is of finite dimension.  There will be multiple ways to define LD, of which we will discuss only a few models. It will be apparent in our discussion that many results exhibit a strong dependence on the specific LD definition chosen. 
\subsection{Boltzmann-von Neumann 
Logarithmic Derivative operator (or BvN-LD for short)}\label{bvn} 
\par Given a parametrised family of states $\rho_\theta,\, \theta \in I$, the BvN prescription for the entropy leads to the Likelihood Operator (LO) 
\begin{equation}\label{e4.1}
 L(\theta) \equiv \ln \rho_\theta \equiv \sum_k \ln \lambda_k(\theta) P_k(\theta),
\end{equation}
which is a self-adjoint operator. The Logarithmic Derivative (LD) operator is given by the derivative of $L$ as 
\begin{equation}\label{e4.2}
H(\theta) = \frac{\dd \ln \rho }{\dd \theta} = \sum_k \frac{\lambda_k'}{\lambda_k} P_k + \sum_k \ln \lambda_k P_k' \equiv H_1(\theta) + H_2(\theta),
\end{equation}
where $H_1(\theta)$ contains the projections $\{P_k\}$ only, and $H_2(\theta)$ contains the family $\{P_k'\}$. The last equality, coming from spectral decomposition of $\rho$,  is valid under the stronger hypothesis of real analyticity of $\theta \mapsto \rho_\theta$. 
\par Another way to derive $H$ is to look at this as a solution $H$ of the Kubo–Mori–Bogoliubov (KMB) equation given in literature in an implicit form (see \cite{haya02, hiai-petz})
\begin{equation}\label{kmb}
\rho' = \int_0^1 \dd t~~ \rho^t H \rho^{1-t},
\end{equation}
and as an integral representation in terms of  $\rho'$ given below  as in \cite{hiai-petz, ruskai-23},
\begin{equation}
H   = \int_0^\infty \dd u  ~\frac{1}{\rho +u} \rho' \frac{1}{\rho +u}.\label{kmb2}
\end{equation}
These forms are also extensively used in literature \cite{AN, PT, haya02, hiai-petz}. It turns out that the solution of the above equation \eqref{kmb} precisely corresponds to the expression derived from the spectral decomposition. This result is demonstrated in the following theorem for finite-dimensional spaces and remains valid in infinite-dimensional settings under certain additional constraints.
 \begin{theorem} \label{kmb-1}
If $\dim \mathcal{H}< \infty$, then $\mathcal{H}(\theta)$, defined by equation \eqref{e4.2} using the spectral decomposition of $\ln \rho$, satisfies the KMB condition given by equation \eqref{kmb}. Conversely, if $H(\theta)$ is defined by equation \eqref{kmb2}, it satisfies equation \eqref{e4.2} derived from spectral decomposition. Therefore, these two expressions are inverses of each other. 
 \end{theorem}
\begin{proof}
Let $\dim \mathcal{H} = N$. Spectral decomposition gives $\rho' = \sum_{j=1}^N \lambda_j' P_j + \sum_{j=1}^N \lambda_j P_j'$. Now, putting explicit form of $H$ in the right hand side of the  equation \eqref{kmb}, one gets 
\begin{align*}
\rho' &=\int_0^1 \dd t ~\left(\sum_{k=1}^N \lambda_k^t P_k\right) \left( \sum_{j=1}^N \frac{\lambda_j'}{\lambda_j} P_j +  \sum_{j=1}^N  \ln \lambda_j P_j' \right)\left(\sum_{l=1}^N \lambda_l^{1-t} P_l\right)\\
&= \sum_{k=1}^N \lambda_k' P_k + \sum_{\substack{j, k,l\\ k \neq l}} \frac{\lambda_k - \lambda_l}{ \ln \lambda_k - \ln \lambda_l } \ln \lambda_j P_k P_j' P_l.
\end{align*}
Observe that the first term of the given equation aligns with the corresponding term in the expression for $\rho'$ involving solely derivatives of eigenvalues $\lambda_k$'s. 
Consequently, the second summation can be restricted to the terms where $k \neq l$. Therefore, it suffices to demonstrate the equality of the second terms, involving the derivatives of the eigenprojections. This second term can be further decomposed into a sum of two distinct components as: 
\[ \sum_{\substack{j, k,l\\ k \neq l}} \frac{(\lambda_k - \lambda_l)\ln \lambda_j}{ \ln \lambda_k - \ln \lambda_l }  P_k P_j' P_l = \sum_{\substack{\text{either }j = k \neq l\\ \text{or } j =l \neq k }} \frac{(\lambda_k - \lambda_l)\ln \lambda_j}{ \ln \lambda_k - \ln \lambda_l }  P_k P_j' P_l  + \sum_{\substack{l \neq j \neq k\\ k \neq l}} \frac{(\lambda_k - \lambda_l)\ln \lambda_j}{ \ln \lambda_k - \ln \lambda_l }  P_k P_j' P_l.\]
Condition on the second sum of the right hand side effectively means that all indices $j, \, k, \, l$ are different. Thus  $P_k P_j' P_l = - P_k' P_j P_l =0$ by using Lemma \ref{l1} and the fact that $P_j$ and $P_l$ are orthogonal, this sum is zero. Now, the first sum can further be simplified  as 
\begin{align*}
&\sum_{\substack{\text{either }j = k \neq l\\ \text{or } j =l \neq k }} \frac{(\lambda_k - \lambda_l)\ln \lambda_j}{ \ln \lambda_k - \ln \lambda_l }  P_k P_j' P_l = \sum_{\substack{k,l\\ k \neq l }} \frac{(\lambda_k - \lambda_l)\ln \lambda_k}{ \ln \lambda_k - \ln \lambda_l }  P_k P_k' P_l + \sum_{\substack{k,l \\ k \neq l }} \frac{(\lambda_k - \lambda_l)\ln \lambda_l}{ \ln \lambda_k - \ln \lambda_l }  P_k P_l' P_l \\
=& \sum_{\substack{k,l\\ k \neq l }} \frac{(\lambda_k - \lambda_l)\ln \lambda_k}{ \ln \lambda_k - \ln \lambda_l }  P_k' P_l + \sum_{\substack{k,l \\ k \neq l }} \frac{(\lambda_k - \lambda_l)\ln \lambda_l}{ \ln \lambda_k - \ln \lambda_l }  (-P_k' P_l) 
= \sum_{\substack{k,l\\ k \neq l }} \frac{(\lambda_k - \lambda_l)(\ln \lambda_k -\ln \lambda_l )}{ \ln \lambda_k - \ln \lambda_l } P_k' P_l \\
=& \sum_{k} \lambda_k P_k' \sum_{l \neq k} P_l - \sum_l \lambda_l  \sum_{k \neq l} P_k' P_l
= \sum_{k} \lambda_k P_k' (I- P_k) + \sum_l \lambda_l P_l'P_l= \sum_{k} \lambda_k P_k'.
\end{align*}
The last step is arrived at on using the fact that $\sum_{k\neq l} P_k' =(I -P_l)' = -P_l'$. 

\par Conversely, given equation \eqref{kmb2} and spectral decomposition, it can be shown that 
\begin{align*}
&\int_0^\infty (\sum_k \frac{1}{\lambda_k+t}P_k) (\sum_l \lambda_l' P_l + \sum_l \lambda_l P_l') (\sum_j \frac{1}{\lambda_j+t} P_j)\dd t \\
=& \sum_k \frac{\lambda_k'}{\lambda_k} P_k +  \sum_l \sum_{k \neq j} \lambda_l P_k P_l' P_j \int_0^\infty \frac{\dd t}{(\lambda_k +t)(\lambda_j +t)} = \sum_k \frac{\lambda_k'}{\lambda_k} P_k  + \sum_l \sum_{k \neq j}  \lambda_l \frac{\ln\frac{\lambda_j}{\lambda_k}}{\lambda_j - \lambda_k} P_k P_l' P_j\\
=& \sum_k \frac{\lambda_k'}{\lambda_k} P_k  + \sum_{k \neq j} \lambda_k \frac{\ln\frac{\lambda_j}{\lambda_k}}{\lambda_j - \lambda_k}  P_k'P_j + \sum_{k \neq l} \lambda_l \frac{\ln\frac{\lambda_j}{\lambda_k}}{\lambda_l - \lambda_k} (-P_k'P_l)
= \sum_k \frac{\lambda_k'}{\lambda_k} P_k + \sum_{k \neq l} \ln \frac{\lambda_k}{\lambda_l} P_k' P_l 
\end{align*}
Here we notice that the  other possibilities involve $P_kP_l'P_j = - P_k' P_l P_j =0$ when $l\neq j,k$ and $k \neq j$. Thus, using the above expression 
\begin{align*}
\frac{\dd \ln \rho}{\dd \theta} &= \sum_k \frac{\lambda_k'}{\lambda_k} P_k + \sum_{k \neq l} \ln \frac{\lambda_k}{\lambda_l} P_k' P_l \\ 
&= \sum_k \frac{\lambda_k'}{\lambda_k} P_k +  \sum_k \ln \lambda_k P_k' (\sum_l P_l) - \sum_k P_k' \sum_l \ln\lambda_l P_l\\
&= \sum_k \frac{\lambda_k'}{\lambda_k} P_k  + \sum_k \ln \lambda_k P_k' \quad \text{ since } \sum_k P_k =I.
\end{align*}
\end{proof}
The next theorem gives a set of criteria under which the LD operator $H$ in the BvN scenario in an infinite dimensional Hilbert space will be a bounded operator and satisfy the KMB equation \eqref{kmb}.
\begin{theorem}
Let $\dim\mathcal{H}= \infty$, and let $\{\rho_\theta\}$ be a family of states as before with its spectral decomposition. Assume furthermore that $\left\{\left| \frac{\lambda_k'}{\lambda_k}\right|\right\}_k$ is bounded and $\left\{ |\ln \lambda_k| \|e_k'\|\right\}_k$ is square-summable. Then 
\begin{enumerate}[(i)]
\item the hypothesis of the Theorem 3.4 are satisfied so that $\rho'$ is in $\mathcal{B}_1(\mathcal{H})$ and 
\[\rho'=\sum_k \lambda_k' P_k + \sum_k \lambda_k P_k';\]
\item $\rho_\theta$ is strongly differentiable with respect to $\theta$ and  
\begin{equation}
H(\theta)= \sum_k \frac{\lambda_k'}{\lambda_k} P_k  + \sum_k \ln \lambda_k P_k' \in \mathcal{B(H)};
\end{equation}
\item $H(\theta)$ in (ii) satisfies the KMB equation \eqref{kmb}. 
\end{enumerate}
\end{theorem}
\begin{proof}
(i) Since $| \lambda'_k |    = \left| \frac{\lambda_k'}{\lambda_k}\right| .\lambda_k \le C \lambda_k$ , it follows that  $\sum _k \|\lambda'_k\| <\infty$.   Also it is clear that  since $\lambda_k \to  0$ as $k \to \infty$,  $| \ln \lambda_k|  >$ some positive  constant $C'$   for  large $k$. This implies by Cauchy-Schwarz inequality that   $\sum_k \lambda_k \| e'_k \|  < {C'}^{-1}  \sum_k   \lambda_k |\ln \lambda_k |  \|e'_k \|  < \infty$, leading to the required result.
\par(ii) By the hypotheses of this theorem,  $\sum_k \frac{\lambda_k'}{\lambda_k}  P_k$  is in $\mathcal{B(H)}$, since the  sequence $\left\{ \frac{\lambda_k'}{\lambda_k}\right\}$ is a bounded sequence. As in the proof of the Theorem 3.4, the second term in the equation (3.8) for $H$   can be estimated as follows: for $f \in \mathcal{H}$ ,    $\sum_k \ln \lambda_k P_k'   = \sum_k  \ln \lambda_k ( \langle e_k,f \rangle  e'_k +  \langle e'_k,f \rangle  e_k)$, 
of which the square of the norm of the second expression $= \sum_k  |  \ln \lambda_k  \langle e'_k,f \rangle |^2   \le ( \sum_k   | \ln \lambda_k \|e'_k \| |^2  ) \|f\|^2$.  
The norm of the other  term  is $\le \sum_k |  \ln \lambda_k| \| e'_k \|  |\langle e_k,f \rangle |   \le \| f \| ( \sum_k | \ln \lambda_k \|e'_k \| |^2 )^{\frac{1}{2}}$. 
This proves that 
the LD operator $H$  is bounded  and   admits an estimate:  
\[ \|H(\theta) \| \le  \sup_k \left\{\left| \frac{\lambda_k'}{\lambda_k} \right|  \right\}  +  2 ( \sum  _ k  | \ln \lambda_k  \|e'_k\| |^2 )^{\frac{1}{2}}.\]

\par (iii)   Set  $\rho^{(N)}  =  \sum_{k=1}^N \lambda_k P_k$   and $H^{(N)}  = \sum_{k=1}^N \frac{\lambda_k'}{\lambda_k} P_k  + \sum_{k=1}^N \ln \lambda_k P_k'$. Then
 $\rho^{(N)}  \to  \rho$  in $\mathcal{B}_1(\mathcal{H})$-norm    and by (ii)    $H^{(N)} \to H$   strongly .  This implies that for each $t \in (0,1]$,  
        ${\rho^{(N)}}^t  \to  \rho^t$   in $\mathcal{B}_{t^{-1}} (\mathcal{H})$ -  norm where $\mathcal{B}_p(\mathcal{H})$ is Schatten - von Neumann ideal in $\mathcal{B(H)}$.   Furthermore,  note that    
         $\| ( I - \rho^{t-s} ) f ||^2  =   \sum_k (1-\lambda^{t-s} )^2 \| P_k f \|^2  \to 0$  as $t \to s$, by the dominated convergence theorem, leading to   the   result   that   
         $\| \rho^t  - \rho^s \|_{s^{-1}}    = \| \rho^s ( I - \rho^{t-s})\|_{s^{-1}}  \to 0$  as $t>s$  decreases to $s$, since $\rho^s$ is in   $\mathcal{B}_{s^{-1}}(\mathcal{H})$.    
     We also observe that for $t=0$ or $1$ respectively, the integrant becomes $H\rho$ or $\rho H$. Putting all these together, we conclude that  $\lim_{N \to \infty} {\rho^{(N)}}^t H^{(N)} {\rho^{(N)}}^{1-t}~\dd t$ converges to $\int_0^1 \rho^t H \rho^{1-t}~\dd t$ in $\mathcal{B}_1(\mathcal{H})$-norm, uniformly with respect to
     $t \in [0,1]$,
\end{proof}

\section{Other Logarithmic Derivative Operators (LD)}\label{hels}
\par In the preceding section, one of the natural definitions of LD, $H(\theta)$ was given, and its properties were studied. However, the passage from $\rho'$ to the LD of $\rho$ can have a few other possibilities, keeping in mind the formal self-adjoint requirement of the LD, of which we consider only three.  Some of these have been studied in the physics literature \cite{hell-book, holevo1, hayashi2}. 
\subsection{Various models of LDs}
\begin{itemize}
\item{(LD$_1$):}  $\wt{H}(\theta) \equiv \frac{1}{2} [\rho^{-1} \rho' + \rho' \rho^{-1}], $ which in finite-dimensional $\mathcal{H}$, is well-defined since $\ker \rho_\theta =\{0\}, \forall \theta $ by hypothesis. However, if $\dim \mathcal{H}$ is infinite, then this needs careful study, and some results in this direction are given later. As before, the Likelihood Operator (LO) is given by $\wt{L}(\theta) =\int^\theta \wt{H}(\tau) \dd \tau$. 
\item{(LD$_2$):}  Another possible symmetric form of LD is given as $\wh{H}(\theta) \equiv \rho^{-\frac{1}{2}} \rho' \rho^{-\frac{1}{2}}$, which again makes perfect sense when $\dim \mathcal{H} < \infty$. 
\item{(LD$_3$):} Symmetric Logarithmic Derivative or SLD for short \cite{hels-73}. Unlike the previous models, this LD is implicitly defined as the solution of the equation 
\begin{equation}\label{sld}
\frac{1}{2}(\wt{\wt{H}}\rho + \rho \wt{\wt{H}}) =\rho'.
\end{equation}
In this context, we have the following easy result.  
\end{itemize}

\begin{theorem}\label{th4.1}
\begin{enumerate}[(i)]
\item If $\dim \mathcal{H} < \infty$, then the equation \eqref{sld} has a unique solution given by 
\begin{equation} \label{sldsol}
\wt{\wt{H}}(\theta) = \int_0^\infty \dd t ~~e^{-t \frac{\rho}{2}} \rho_\theta' e^{-t \frac{\rho}{2}}. 
\end{equation}
\item If $\dim \mathcal{H} < \infty $, then in all expressions for LDs:
\[ \mathbb{E}_\theta(\{ H,\, \wt{H}, \, \wh{H},\, \wt{\wt{H}} \}) \equiv \tr (\rho \{ H,\, \wt{H}, \, \wh{H},\, \wt{\wt{H}} \})=0,\]
in each case. 
\end{enumerate}
\end{theorem}
\begin{proof}
\emph{Part (i): } First, since $\{ e^{-t \frac{\rho}{2}}: t \ge 0\}$ is a norm-continuous semigroup with bounded generator $-\frac{1}{2} \rho$, 
\begin{equation}\label{ee1} \frac{1}{2}(\wt{\wt{H}}\rho + \rho \wt{\wt{H}})  = -\int_0^\infty \dd t ~\frac{\dd}{\dd t} \left( e^{-t \frac{\rho}{2}} \rho' e^{-t \frac{\rho}{2}} \right) = \rho'.\end{equation}

\par \emph{Part (ii): } By equation (3.6),  $\tr \rho H = \int_0^1 \dd t ~ \tr \left( \rho^t H \rho^{1-t} \right) = \tr \left( \int_0^1 \dd t~ \rho^t  H \rho^{1-t} \right) = \tr \rho' =0$. 
\par Similarly for $\wt{\wt{H}}$, taking trace in both sides of the \eqref{ee1}, one can obtain $\tr(\rho \wt{\wt{H}}) =\tr \rho' =0$. The other forms are equally simple. 
\end{proof}
\begin{rem}
It is clear that all the three LD-operators defined above are unbounded operators if the $\dim \mathcal{H}$ is infinite, for in such cases, $\lambda_k \to 0$ as $ k \to \infty$. In particular, if $\rho e_k = \lambda_k e_k$, then $\langle e_k, H e_k \rangle  = \frac{1}{2\lambda_k} \langle e_k, \rho' e_k \rangle $ so that $| \langle e_k, H e_k \rangle |  \to \infty$ as $k \to \infty$. 
\end{rem}

We conclude this with an observation that the assumption of commutativity of $\rho$ and $\rho'$ gives a striking result. All four models defined above give the same LD operator as those defined earlier for finite dimensions. 
\begin{theorem}\label{th3.1}
Let $\dim \mathcal{H} < \infty$. If  $\rho$ commutes with $\rho'$, then
$H(\theta)=\wt{H}(\theta)=\wh{H}(\theta) = \wt{\wt{H}}(\theta)$ for all $\theta$
\end{theorem}
\begin{proof}
If $\rho_\theta = \sum \lambda_k (\theta) P_k(\theta)$, and if $\rho$ commutes with $\rho'$, then $H(\theta) = \rho' \rho^{-1} = \sum_k \dfrac{\lambda_k'}{\lambda_k} P_k$, since by  Lemma \ref{l1} part II, $\sum_k \lambda_k P_k' =0$, i.e. $\rho'= \sum_k \lambda_k' P_k$. Finally note that 
\[\wh{H} =\wt{H} = \wt{\wt{H}} = \rho' \rho^{-1} =\sum_k \dfrac{\lambda_k'}{\lambda_k} P_k =H.\] 
\end{proof}
\subsection{Cram{\'e}r-Rao bound for finite dimension}\label{s4}

\par The next goal involves proving the Cramér-Rao (CR) bound for all four quantum logarithmic derivative (LD) scenarios previously described, assuming $\dim \mathcal{H} < \infty$. 
 Here, the LD operators are always well-defined. The \emph{Quantum Fisher Information (QFI)} in the scenarios $LD_1$, $LD_2$, and $LD_3$ are defined as the variance of corresponding LD operators and are denoted by $\wt{\bm{I}}, \, \wh{\bm{I}}$, and $\wt{\wt{\bm{I}}}$ respectively. Since by Theorem \ref{th4.1}, $\mathbb{E}(\wt{H}) = \mathbb{E}(\wh{H}) = \mathbb{E}(\wt{\wt{H}}) =0$, $\wt{\bm{I}}(\theta) = \var \wt{H} = \tr \rho \wt{H}^2$,  $\wh{\bm{I}}(\theta) = \var \wh{H} = \tr \rho \wh{H}^2$, and $\wt{\wt{\bm{I}}}(\theta) = \var \wt{\wt{H}} = \tr \rho \wt{\wt{H}}^2$. We postpone the discussion of QFI of the BvN scenario later since that requires a modified definition of variance (see Theorem \ref{th4.5}).

 \par Before we proceed, we note an important corollary of  Theorem \ref{th3.1}.  
\begin{cor}\label{cor4.4}
Let $ \dim \mathcal{H} < \infty$. If $\rho$ commutes with $\rho'$ then, $H(\theta) = H_1(\theta)$ in all versions and 
\[ \wt{\bm{I}}(\theta) = \wh{\bm{I}}(\theta) = \wt{\wt{\bm{I}}}(\theta) = \tr\left[ \sum_k \frac{(\lambda_k')^2}{\lambda_k} P_k(\theta)\right]. \]
\end{cor}
\begin{proof}
This follows from Theorem \ref{th3.1}, as $\wh{H} =\wt{H} = \wt{\wt{H}} $. 
\end{proof}

\subsubsection{Quantum CR bound for LD$_1$} \label{ss:4.1}
In this  scenario,  for a single system for a Hilbert space
$\mathcal{H}$, we have 
\[2\wt{H}(\theta) = \rho^{-1} \rho' + \rho' \rho^{-1},\]  and let  $\Theta$ be an 
unbiased estimator operator for $\theta$, i.e.
$\mathbb{E}_\theta(\Theta) =\theta$, and we also assume that in all
cases, the estimator operator $\Theta$ is not explicitly dependent on
the parameter $\theta$. On differentiating this with respect to
$\theta$ and using the fact that  $\tr \rho_\theta'=0$ we get  
\begin{equation}\label{e5.0}
\tr(\rho_\theta' (\Theta - \theta) ) =1.
\end{equation}

\par Next we need to solve $\rho'$ in terms of $\wt{H}$ in the
equation
\begin{equation} \label{ly}
2\wt{H} = \rho^{-1} \rho' + \rho' \rho^{-1}.
\end{equation}
$\rho^{-1}$ is the 
generator of a (strongly) dissipative  semi-group
$\left\{\exp\left(-\frac{t}{\rho}\right)\right\}_{t \ge 0}$, and note that 
$\left\|\exp\left(-\frac{t}{\rho}\right)\right\| \le
\exp\left(-\frac{t}{\lambda_1}\right)$. Next, we need to solve for
$\rho'$ in terms of $\wt{H}$ in \eqref{ly}, and since this is
a version of the Lyapunov problem, we expect the solution to be
given by 
\begin{equation} \label{eq4.1}
\rho' = \int_0^\infty \exp\left(-\frac{t}{\rho}\right) (2\wt{H}) \exp\left(-\frac{t}{\rho}\right) \dd t.
\end{equation}
Even if $\dim \mathcal{H} =\infty$, for  $f,\,g \in D(\rho^{-1})$ the equation \eqref{ly} takes the form 
\begin{eqnarray*}
&&\int_0^\infty \left\langle  \exp\left(-\frac{t}{\rho}\right)f,  (2\wt{H}) \exp\left(-\frac{t}{\rho}\right) g \right\rangle \dd t\\ 
&=& \int_0^\infty \left\langle  \exp\left(-\frac{t}{\rho}\right)f,
(\rho^{-1} \rho' + \rho' \rho^{-1}) \exp\left(-\frac{t}{\rho}\right) g
\right\rangle \dd  t \\
&=&  \int_0^\infty-\frac{\dd}{\dd t} \left \langle
\exp\left(-\frac{t}{\rho}\right)f, \rho'
\exp\left(-\frac{t}{\rho}\right)g \right \rangle \dd t \\
&=& - \left \langle  \exp\left(-\frac{t}{\rho}\right)f, \rho'
\exp\left(-\frac{t}{\rho}\right)g \right \rangle \bigg|_{t=\infty}+
\left\langle f, \rho' g\right\rangle  =\left\langle f, \rho'
g\right\rangle.
\end{eqnarray*}
This shows that the densely-defined quadratic form on $D(\rho^{-1}) \times D(\rho^{-1})$, such that the right hand side of equation \eqref{eq4.1} has an extension as a $\mathcal{B(H)}$ element satisfying \eqref{eq4.1}. (More on this is in Section \ref{s5}.) 
This proves the equation \eqref{eq4.1} and  implies from \eqref{e5.0} that (for $\dim\mathcal{H} < \infty$) 
\begin{eqnarray*}
1 = \tr \left[ \rho'(\Theta -\theta)\right]&=& \int_0^\infty \tr \left[
\exp\left(-\frac{t}{\rho}\right) (2\wt{H}) \exp\left(-\frac{t}{\rho}\right)
(\Theta -\theta)\dd t \right]\\
&=& 2 \int_0^\infty \left\langle \wt{H} \exp\left(-\frac{t}{\rho}\right),
\exp\left(-\frac{t}{\rho}\right) (\Theta -\theta)
\right\rangle_{\mathcal{B}_2} \dd t.
\end{eqnarray*}
Since this  is a positive number, using the Cauchy-Schwartz inequality
for the inner product in $\langle\cdot,\cdot\rangle_{\mathcal{B}_2}$
in ${\mathcal{B}_2}$ and the integral, we get that   
\begin{eqnarray*}
1 &=& 4 \left| \int_0^\infty \left\langle \wt{H}
\exp\left(-\frac{t}{\rho}\right), \exp\left(-\frac{t}{\rho}\right)
(\Theta -\theta) \right\rangle_{\mathcal{B}_2} \dd t \right|^2 \\
& \le & 4 \int_0^\infty  \left\|   \wt{H}
\exp\left(-\frac{t}{\rho}\right) \right\|_2^2 \dd t ~~  \int_0^\infty
\left\|\exp\left(-\frac{t}{\rho}\right) (\Theta -\theta) \right\|_2^2
\dd t \\
& \le & 4\int_0^\infty \tr \left[ \exp\left(-\frac{t}{\rho}\right)
\wt{H}^2 \exp\left(-\frac{t}{\rho}\right) \right]\dd t
~~~\int_0^\infty \tr \left[ \exp\left(-\frac{t}{\rho}\right) (\Theta
-\theta)^2 \exp\left(-\frac{t}{\rho}\right) \right] \dd t   \\
&=& 4 ~~ \tr \left[\wt{H}^2 \int_0^\infty
\exp\left(-\frac{2t}{\rho}\right) \dd t \right] ~~\tr \left[ (\Theta
-\theta)^2 \int_0^\infty \exp\left(-\frac{2t}{\rho}\right) \dd t
\right]    
\end{eqnarray*}
By the spectral theorem, $\int_0^\infty
\exp\left(-\frac{2t}{\rho}\right) \dd t = \frac{1}{2} \rho$. As a
result, the above inequality leads to  
\begin{eqnarray*}
1 &\le & \tr\left( \rho \wt{H}(\theta)^2 \right) \tr \left(\rho(\Theta
- \theta)^2 \right)
= \wt{\bm{I}}(\theta)\var_\theta (\Theta) 
\quad \text{or, } ~~ \var_\theta (\Theta) \ge \frac{1}{\wt{\bm{I}}(\theta)}.
\end{eqnarray*}

\par Now we move on to a compound $n$-system.  Denote $ \mathfrak{H} =
\mathcal{H}^{\otimes n}$ as the Hilbert space of $n$-copies of Hilbert
space $\mathcal{H}$ of a single system. Then $n$ independent observables $X_1, \,
X_2, \cdots, X_n$ can be described in  $\mathfrak{H}$ as 
\[ \mathfrak{X}_j = \underbrace{I \otimes \cdots \otimes  I}_{j-1}
\otimes X_j \otimes \underbrace{I \otimes \cdots \otimes I}_{n-j+1},
\quad 1 \le j \le n.\]
Let  $\Theta \equiv \Theta(\mathfrak{X}_1, \cdots, \mathfrak{X}_n)$
(assumed to be independent of $\theta$)  in $\mathfrak{H}$ be the
estimator of $\theta$, and furthermore let $\Theta$  be an unbiased
estimator of $\theta$ in the state $\bm{\varrho}_\theta =
{\rho_1}_\theta \otimes \cdots \otimes {\rho_n}_\theta$  (with
${\rho_j}_\theta$ a state of single particle for each $j$)  in
$\mathfrak{H}$, i.e.
\begin{eqnarray}\label{e5.1}
\mathbb{E}_\theta (\Theta(\mathfrak{X}_1, \cdots, \mathfrak{X}_n)) &=&
\tr_{\mathfrak{H}} \left(\bm{\varrho}(\theta)  \Theta(\mathfrak{X}_1,
\cdots, \mathfrak{X}_n)\right) =\theta, \quad \text{ or } \\
\tr \bm{\varrho}(\theta) (\Theta -\theta) &=&0. \nonumber
\end{eqnarray}
Since each $\rho_j(\theta)$ is assumed to be differentiable with
respect to $\theta$, the compound state $\bm{\varrho}(\theta)$ is also
differentiable. Differentiating  \eqref{e5.1} with respect to $\theta$,  we get
\begin{equation}\label{cr1}
1 \equiv \tr \left[ \bm{\varrho}'(\theta) \Theta(\mathfrak{X}_1, \cdots, \mathfrak{X}_n)\right],
\end{equation}
where
$\bm{\varrho}'= \sum_{j=1}^n \rho_1 \otimes \cdots \rho_{j-1} \otimes
\rho_j' \otimes \rho_{j+1} \otimes \cdots \otimes \rho_n.$
We set 
\begin{equation}\label{crld1}
\wt{\bm{H}}(\theta) = \sum_{j=1}^n \underbrace{I\otimes
\cdots\otimes I}_{j-1 \text{ times}} \otimes \wt{H}_j (\theta) \otimes
\underbrace{I \otimes \cdots\otimes I}_{n-j \text{ times}}, 
\end{equation}
and  $\bm{\varrho}^{-1} = \rho_1^{-1} \otimes \cdots \otimes
\rho_n^{-1}$ (recalling that the $\rho_j$'s are invertible), and as before, we need to solve $\bm{\varrho}'$ in terms of
$\wt{\bm{H}}$ in the equation
$\wt{\bm{H}} = \frac{1}{2^n}(\bm{\varrho}^{-1} \bm{\varrho}' +
\bm{\varrho}' \bm{\varrho}^{-1}).$
Denote $\bm{t}$ as the vector $\bm{t} =(t_1, \cdots, t_n)$, the
multiple-integral $\int_0^\infty \cdots \int_0^\infty \dd t_1 \cdots
\dd t_n$ by $\int_{\bm{0}}^{\bm{\infty}} \bm{\mathrm{d}t}$, and
$\exp\left(-\frac{\bm{t}}{\bm{\varrho}}\right) =
\exp\left(-\frac{t_1}{\rho_1}\right) \otimes \cdots \otimes
\exp\left(-\frac{t_n}{\rho_n}\right)$, and 
consider the multiple integral:
\begin{eqnarray*}
 && \int_{\bm{0}}^{\bm{\infty}} \bm{\mathrm{d}t}
\exp\left(-\frac{\bm{t}}{\bm{\varrho}}\right) \left( 2^n \bm
{\wt{H}}(\theta)  \right)
\exp\left(-\frac{\bm{t}}{\bm{\varrho}}\right) \\
&=& \int_{\bm{0}}^{\bm{\infty}} \bm{\mathrm{d}t} \left[\sum_{j=1}^n
2\exp\left(-\frac{2t_1}{\rho_1}\right) \otimes \cdots \otimes
2\exp\left(-\frac{2t_{j-1}}{\rho_{j-1}}\right) \otimes
\exp\left(-\frac{t_j}{\rho_j}\right)2\wt{H}_j
\exp\left(-\frac{t_j}{\rho_j}\right) \right.\\
&& \hspace{1cm}\left.\otimes 2\exp\left(-\frac{2t_{j+1}}{\rho_{j+1}}\right)
\otimes \cdots \otimes 2\exp\left(-\frac{2t_n}{\rho_n}\right) \right]\\
&=& \sum_{j=1}^n \rho_1 \otimes  \cdots \otimes \rho_{j-1} \otimes
\rho_j' \otimes \rho_{j+1} \otimes \cdots \otimes \rho_n 
= \bm{\varrho}',
\end{eqnarray*}
since  $ \int_0^\infty \dd t_j~~ \exp
\left(-\frac{t_j}{\rho_j}\right) 2\wt{H}_j \exp
\left(-\frac{t_j}{\rho_j}\right)  =\rho_j'$ by equation \eqref{eq4.1}. Next, differentiating the equation \eqref{e5.1} with respect to $\theta$, we are led to  
 $\tr_\mathfrak{H} \bm{\varrho}'(\Theta - \theta) =1$, since
 \newline
$\tr_\mathfrak{H} \bm{\varrho}' = \tr_\mathfrak{H} \sum_{j=1}^n
\rho_1 \otimes \cdots \otimes \rho_{j-1} \otimes \rho_j' \otimes
\rho_{j+1} \otimes \cdots \otimes \rho_n = \sum_{j=1}^n
\tr_{\mathcal{H}_j} \rho_j' =0.$ 
Using this, we get 
\begin{eqnarray*}
1&=&\tr_\mathfrak{H}\left\{ \left(\int_{\bm{0}}^{\bm{\infty}} \bm{\mathrm{d}t}
\exp\left(-\frac{\bm{t}}{\bm{\varrho}}\right) \left( 2^n \bm
{\wt{H}}(\theta)  \right) 
\exp\left(-\frac{\bm{t}}{\bm{\varrho}}\right)\right) (\Theta -\theta)\right\}\\
 &=& \int_{\bm{0}}^{\bm{\infty}} \bm{\mathrm{d}t} ~~~2^n \left\langle
\wt{\bm{H}} \exp\left(-\frac{\bm{t}}{\bm{\varrho}}\right),
\exp\left(-\frac{\bm{t}}{\bm{\varrho}}\right) (\Theta -\theta)
\right\rangle_{\mathcal{B}_2(\mathfrak{H})}  \\
&\le & 2^n  \sqrt{\int_{\bm{0}}^{\bm{\infty}} \bm{\mathrm{d}t} \tr \left\{
\exp\left(-\frac{2\bm{t}}{\bm{\varrho}}\right) {\bm{\wt{H}}^2 }\right\}}\cdot
\sqrt{ \int_{\bm{0}}^{\bm{\infty}} \bm{\mathrm{d}t} ~~\tr \left\{(\Theta -\theta)^2
\exp\left(-\frac{2\bm{t}}{\bm{\varrho}}\right) \right\}}.
\end{eqnarray*}
Finally, since 
$\int_{\bm{0}}^{\bm{\infty}} \bm{\mathrm{d}t}
\exp\left(-\dfrac{\bm{2t}}{\bm{\varrho}}\right) =2^{-n}\bm{\varrho}$,
the above inequality  can be written as 
\newline $1 \le \sqrt{ \tr_{\mathfrak{H}}(\bm{\varrho} \wt{\bm{H}}^2)  }
\sqrt{ \tr_{\mathfrak{H}}(\bm{\varrho} (\Theta -\theta)^2)  } .  $ 
This implies the inequality 
\[\var_{\Theta, \mathfrak{H}} (\Theta) \ge \frac{1}{\wt{\bm{I}}^{(n)}
(\theta)},\]
where $\wt{\bm{I}}^{(n)} (\theta) = \tr_{\mathfrak{H}}(\bm{\varrho} \wt{\bm{H}}^2)$, the QFI for compound $n$-system. 
If the observables $X_1, \cdots, X_n$ are distributed identically,
i.e. $\rho_j=\rho$ for $j=1, \cdots, n$, then $\wt{\bm{H}} =
\sum_j\underbrace{I \otimes \cdots \otimes  I}_{j-1}
\otimes\wt{H}\otimes \underbrace{I \otimes \cdots \otimes I}_{n-j} $,
which implies $\wt{\bm{I}}^{(n)} (\theta)=n\wt{\bm{I}} (\theta)$, and hence the
above inequality takes the form,
\[\var_{\theta, \mathfrak{H}} (\Theta) \ge \frac{1}{n\wt{\bm{I}} (\theta)}.\]

\subsubsection{CR-bound for LD$_2$: } \label{ss:4.3}

\par In this scenario, if we set the likelihood estimator 
$\wh{L}(\theta) = \int^\theta \wh{H}(\theta)~\dd \theta,$
with 
$ \wh{H}(\theta) = {\rho}^{-\frac{1}{2}} \rho' \rho^{-\frac{1}{2}},$ 
then \eqref{e5.0} implies that 
\[1 = \tr\left[ \rho^{\frac{1}{2}}  \wh{H}(\theta)  {\rho}^{\frac{1}{2}}
(\Theta -\theta)\right] = \left\langle  \wh{H}(\theta)
{\rho}^{\frac{1}{2}},  {\rho}^{\frac{1}{2}} (\Theta
-\theta)\right\rangle_{\mathcal{B}_2},\]
 which by Cauchy-Schwartz inequality gives 
 \begin{eqnarray*}
 1\le \left\| H(\theta) {\rho}^{\frac{1}{2}} \right\|_2 ~ \left\|
{\rho}^{\frac{1}{2}} (\Theta -\theta)\right\|_2 
 &=& \left[\tr \left(  {\rho}^{\frac{1}{2}}  \wh{H}(\theta)
{\rho}^{\frac{1}{2}} \right)\right]^{\frac{1}{2}} \left[\tr \left(
{\rho}^{\frac{1}{2}}  (\Theta -\theta)  {\rho}^{\frac{1}{2}}
\right)\right]^{\frac{1}{2}}\\
 &=& \left(  {\mathbb{E}}_\theta ( \wh{H}(\theta)^2)\right)^{\frac{1}{2}}
\left(\var_{\theta} (\Theta)\right)^{\frac{1}{2}}, \quad \text{ or }
\\
\var_\theta(\Theta) &\ge & \frac{1}{  {\mathbb{E}}_\theta \left(
\wh{H}(\theta)^2\right)}= \frac{1}{\wh{\bm{I}}(\theta)},
 \end{eqnarray*}
 which is the QCR bound in this scenario for a single quantum system. 
  
\par The $n$-system from  of CR-bound can also be calculated directly
using  the structure given in \S\ref{ss:4.1}. Let 
$ \wh{\bm{H}} = \bm{\varrho}^{-\frac{1}{2}}\bm{\varrho}'
\bm{\varrho}^{-\frac{1}{2}}$ which is equal to $\sum_j\underbrace{I
\otimes \cdots \otimes  I}_{j-1} \otimes \wh{H}_j \otimes \underbrace{I
\otimes \cdots \otimes I}_{n-j} $ and 
\[\bm{\varrho}' =  \bm{\varrho}^{\frac{1}{2}}\wh{\bm{H}}
\bm{\varrho}^{\frac{1}{2}} = \sum_j \rho_1 \otimes \cdots \otimes
\rho_{j-1} \otimes \rho_j^{\frac{1}{2}} \wh{H}_j \rho_j^{\frac{1}{2}}
\otimes \rho_{j+1} \otimes \cdots \otimes \rho_n,\]
 then it is easy to check that
\begin{eqnarray*}
1 = \tr_\mathfrak{H} \bm{\varrho}'(\Theta -\theta)
&=& \tr_\mathfrak{H} \left( \bm{\varrho}^{\frac{1}{2}}\wh{\bm{H}}
\bm{\varrho}^{\frac{1}{2}} (\Theta -\theta)\right)
= \left| \left\langle \bm{\varrho}^{\frac{1}{2}}\wh{\bm{H}},
\bm{\varrho}^{\frac{1}{2}} (\Theta
-\theta)\right\rangle_{\mathcal{B}_2} \right|\\
&\le & \sqrt{\tr_\mathfrak{H}\bm{\varrho}^{\frac{1}{2}}\wh{\bm{H}}^2
\bm{\varrho}^{\frac{1}{2}}} \sqrt{ \tr_\mathfrak{H} (\Theta -\theta)
\bm{\varrho}  (\Theta -\theta)} \quad \text{ by Cauchy-Schwartz,}\\
&=& \sqrt{\tr_\mathfrak{H}(\bm{\varrho}\wh{\bm{H}}^2) ~~\tr_\mathfrak{H}(
\bm{\varrho} (\Theta -\theta)^2 )}.
\end{eqnarray*}
Thus $1 \le \sqrt{ \wh{\bm{I}}^{(n)}(\theta) \var_{\theta, \mathfrak{H}}
(\Theta)}$, where $\wh{\bm{I}}^{(n)}(\theta) =
\tr_\mathfrak{H}(\bm{\varrho}\wh{\bm{H}}^2)$ is the quantum Fisher
information for compound $n$-systems. Thus 
\[\var_{\theta, \mathfrak{H}} (\Theta) \ge
\frac{1}{\wh{\bm{I}}^{(n)}(\theta)}.\]
Furthermore, if the observables $X_1, \cdots, X_n$ are distributed
identically, i.e. $\rho_j=\rho$ so that  $\wh{H}_j=\wh{H}$ for
all $j=1, \cdots, n$ and $\tr_\mathfrak{H} \bm{\varrho} \wh{\bm{H}}^2 = n
\wh{\bm{I}}(\theta)$ and 
\[\var_{\theta, \mathfrak{H}} (\Theta) \ge \frac{1}{n\wh{\bm{I}}(\theta)}.\] 
  
  \subsubsection{CR bound for LD$_3$} 
  In finite dimension, we have seen that the SLD is defined implicitly by the equation  $\wt{\wt{H}} \rho + \rho \wt{\wt{H}} = 2 \rho'$ and on the other hand, if one measures an observable $\Theta$ for the value of $\theta$ such that it is unbiased estimator in the state $\rho_\theta$, then one gets as before, 
  \[\tr \rho' (\Theta - \theta) =1.\]
 Eliminating $\rho'$ in favour of $\wt{\wt{H}}$ one gets 
 \begin{eqnarray*}
1&=& \frac{1}{2} \tr \left[ (\wt{\wt{H}} \rho + \rho \wt{\wt{H}})(\Theta - \theta) \right] \\
2 &=& \tr \left( (\wt{\wt{H}} \rho^{\frac{1}{2}} ) ( \rho^{\frac{1}{2}} (\Theta - \theta)\right ) + \tr \left((\rho^{\frac{1}{2}} \wt{\wt{H}}) ( (\Theta - \theta)  \rho^{\frac{1}{2}}\right) \\
&=& \left\langle \rho^{\frac{1}{2}} \wt{\wt{H}}, \rho^{\frac{1}{2}} (\Theta - \theta) \right\rangle_2 + \left\langle \wt{\wt{H}} \rho^{\frac{1}{2}}, (\Theta -\theta) \rho^{\frac{1}{2}} \right\rangle_2\\
&\le& \left\| \rho^{\frac{1}{2}} \wt{\wt{H}} \right\|_2 \left\| \rho^{\frac{1}{2}} (\Theta - \theta) \right\|_2  + \left\| \wt{\wt{H}} \rho^{\frac{1}{2}}\right\|_2 \left\|  (\Theta -\theta) \rho^{\frac{1}{2}} \right\|_2 = 2 \tr (\rho \wt{\wt{H}}^2) \tr(\rho(\Theta - \theta)^2),
 \end{eqnarray*}
 leading to $\var_\theta (\Theta) \ge \dfrac{1}{\tr (\rho \wt{\wt{H}}^2)} = \dfrac{1}{\wt{\wt{\bm{I}}}(\theta)}$. 
The calculations for $n$-systems follow a similar set of steps as given in the previous computations. 

\subsubsection{Quantum CR bound for BvNLD ($\dim \mathcal{H} < \infty$).} As earlier, we start with the equation \eqref{kmb} connecting the LD operator $H$ with $\rho'$ as $\rho' = \int_0^1 \rho^t H \rho^{1-t} ~\dd t$. 

\begin{theorem}\label{th4.5}
Let $\{\rho_\theta: \theta \in I\} $ be a family of states satisfying Assumption \ref{as3.4}. Then,
\begin{enumerate}[(i)]
\item the variance of the estimator operator $\Theta$ satisfies a quantum CR bound with a modified QFI.
\item Furthermore, the likelihood operator associated with the BvN LD satisfies the maximality property in expectation. 
\item Using the Kullback - Leibler expression for quantum relative entropy, one can also derive the new QFI associated with the BvN LD. 
\end{enumerate}
\end{theorem}

\begin{proof}
(i) The equation  \eqref{e5.0}:  $\tr \rho' (\Theta -\theta) =1,$
leads to 
\begin{align}
1 &= \left| \tr (\rho' (\Theta -\theta)) \right|^2 = \left| \tr \int_0^1 \rho^t H \rho^{1-t} (\Theta - \theta) \dd t \right|^2 \nonumber\\
&= \left| \int_0^1 \tr \left[ \left( \rho^{\frac{t}{2}} H \rho^{\frac{1- t}{2}} \right) \left( \rho^{\frac{1- t}{2}} (\Theta - \theta) \rho^{\frac{t}{2}} \right) \right]~\dd t \right|^2 \nonumber \\
&\le \left[ \int_0^1  \dd t~ \left\| \rho^{\frac{1- t}{2}} H \rho^{\frac{t}{2}} \right\|_2^2 \right] \left[ \int_0^1  \dd t~ \left\| \rho^{\frac{1- t}{2}} (\Theta -\theta) \rho^{\frac{t}{2}} \right\|_2^2 \right] \label{bvncr1}
\end{align}
If we now introduce a different definition of variance of an observable $Y$ in quantum theory, viz. 
\[\breve{V}_\rho(Y) \equiv \int_0^1 \left\| \rho^{\frac{1- t}{2}} (Y - \tr(\rho Y)) \rho^{\frac{t}{2}} \right\|_2^2 ~\dd t,\]
(as has been done in Hiai-Petz \cite{hiai-petz}, and Amari- Nagayoka \cite{AN}, which coincides with the usual definition of variance in commutative theory),  and if we designate the associated definition of QFI: ${\bm{I}}(\theta)\equiv \breve{V}_\rho(H)=  \int_0^1 \dd t \|\rho^{\frac{1-t}{2}} H \rho^{\frac{t}{2}}\|_2^2 $,   since $0= \tr(\rho')
     =   \int_0^1 \tr ( \rho^t H \rho^{1-t} ) =  \tr ( \rho H )$. Then we have the following quantum CR-bound: 
\[ \breve{V}_\rho(\Theta) \ge \breve{{\bm{I}}}(\theta)^{-1}.\]

(ii)  Since by Theorem \ref{kmb-1} the expression $H= \dfrac{\dd \ln\rho_\theta}{\dd \theta}$ satisfies equation \eqref{kmb}: $\rho'= \int_0^1 \dd t~ \rho^t H \rho^{1-t}$, we note that, $0= (\tr \rho)' = \tr \rho' = \tr \int_0^1 \dd t (\rho^t H \rho^{1-t} ) = \tr (\rho H)$. Thus 
\begin{align}
\mathbb{E}_\theta (H')  \equiv \tr (\rho H') = \tr \left( \rho \frac{\dd^2\ln \rho_\theta}{\dd\theta^2}\right)
&=(\tr \rho H)' - \tr(\rho' H) = -\tr(\rho' H) \nonumber \\
&=-\int_0^1 \tr(\rho^t H \rho^{1-t} H) \equiv -I(\theta), \label{4.11}
\end{align}
which is the QFI for the BvN scenario  (as per the new definition above). Thus, \newline $\mathbb{E}_\theta(L'')= -I(\theta)<0$, thereby justifying the terminology of $L$ being a maximal likelihood operator in expectation for the estimation of $\theta$. 

\par 
(iii) Given two states $\rho$ and $\sigma$, the Kullback - Leibler quantum relative entropy (entropy of $\rho$ relative to $\sigma$) \cite{OP} is defined as:  $S_{rel} (\rho\|\sigma)= \tr (\rho(\ln \rho - \ln \sigma))$.  In order that this makes sense (even when $\dim \mathcal{H} = \infty$), we assume as before $\ker \rho = \ker \sigma = \{0\}$. The relation between relative entropy and QFI from the BvN LD scenario is as follows \cite{haya02}:
\[\frac{1}{\vep^2} \tr  \left[\rho_{\theta+ \vep} (\ln \rho_{\theta+ \vep} - \ln \rho_\theta) \right] = \tr \frac{(\rho_{\theta+ \vep} - \rho_\theta)}{\vep} \frac{(\ln \rho_{\theta+ \vep} - \ln\rho_\theta)}{\vep} + \tr \rho_\theta\frac{(\ln \rho_{\theta+ \vep} - \ln\rho_\theta)}{\vep^2}. \]
As $\vep \to 0$, the first term converges to $\tr \rho_\theta' \frac{\dd \ln\rho_\theta}{\dd \theta} = \tr(\rho' H)$. Since $\tr\rho H= \tr \left(\rho \frac{\dd \ln \rho_\theta}{\dd \theta} \right) =0$,
\begin{align*}
\lim_{\vep \to 0} \frac{1}{\vep^2} \tr \left[ \rho_{\theta +\vep} \left( \ln \rho_{\theta +\vep} - \ln \rho_\theta \right)\right] &=\tr (\rho' H) + \lim_{\vep \to 0} \tr \left[\rho_\theta \left\{\frac{(\ln \rho_{\theta+\vep} - \ln \rho_\theta) - \vep \frac{\dd\ln\rho_\theta}{\dd \theta}}{\vep^2} \right\}\right]\\
&= \tr(\rho' H) + \frac{1}{2} \tr \left( \rho \frac{\dd^2 \ln \rho}{\dd \theta^2} \right) = \tr(\rho' H) + \frac{1}{2} \tr(\rho H')\\
&= \tr \left( \int_0^1 \rho^t H \rho^{1-t} H \dd t\right) +  \frac{1}{2} \tr(\rho H') = \frac{1}{2} \bm{I}(\theta),
\end{align*}
where we have used from \eqref{4.11} that $\tr(\rho H') = - \tr (\rho' H) = - \bm{I}(\theta)$.
Thus  
\newline $\lim_{\vep \to 0} \frac{2}{\vep^2} S_{rel} (\rho_{\theta+\vep}\| \rho_\theta) = \bm{I}(\theta)$. 
\end{proof}

\begin{rem}
\begin{enumerate}[(i)]
\item  If $\rho$ commutes with $\rho'$, all the LD's are equal with each other, and therefore $I(\theta)$  equals all other three QFI's as in Corollary \ref{cor4.4}. 
\item For the BvN LD scenario, $\bm{H}$ for the composite $n$-system is 
\[ \bm{H} = \sum_{j=1}^n I \otimes \cdots \otimes I \otimes \underbrace{H}_{j\text{-th position}} \otimes I \otimes \cdots \otimes I,\]
just as in equation \eqref{crld1}, and therefore $\bm{I}^{(n)}(\theta) = n \bm{I}(\theta)$. Exactly the same consideration is valid for the SLD case and both lead to identical estimates.  
\item Thus, it may be mentioned that only with this modified definition of QFI, we have been able to prove the maximality of the likelihood operator for the  BvN LD case. 

\end{enumerate}
\end{rem}

\section{LD and QFI for infinite dimensional case}\label{s5}

\par Some of the results given in the earlier sections, in particular Lemma \ref{l1}, are valid even if
the $\mathcal{H}$ is of infinite dimension. 
However, the definitions of likelihood operator \eqref{e4.1}, logarithmic derivative, and quantum Fisher
Information 
are valid as well, with a caveat that the operators themselves may be
unbounded, even if they exist.

\begin{theorem}
Let $\rho_.: I \ni \theta \mapsto \rho_\theta \in \mathcal{B}_{1,+}(\mathcal{H})$ be weak*-differentiable family of quantum states satisfying Assumption 3.2.  Define the following sesquilinear quadratic forms:
\begin{enumerate}[(i)]
\item $F_{{H}} [u,v]= \int_0^\infty \dd x \left\langle \frac{1}{\rho +x} u, (\rho' + \|\rho'\| +1)\frac{1}{\rho +x} v \right \rangle.$
\item $F_{\wt{H}} [u,v] = \int_0^\infty \dd s \left \langle e^{- \frac{s}{2\rho}} u, (\rho' + \|\rho'\| +1)  e^{- \frac{s}{2\rho}} v \right \rangle. $
\item $F_{\wh{H}}[u,v] = \left\langle \rho^{-\frac{1}{2}} u, (\rho' + \|\rho'\| +1) \rho^{-\frac{1}{2}} v \right \rangle. $
\item 
$ F_{\wt{\wt{H}}} [u,v] = \int_0^\infty \dd s \left \langle e^{-s \frac{\rho}{2}} u, (\rho' + \|\rho'\| +1)  e^{-s \frac{\rho}{2}} v \right \rangle .$

\end{enumerate}
For the cases (i), (iii), and (iv), the form domain is defined to be $D(\rho^{-\frac{1}{2}})$, and for the second case, the form domain is $D(\rho^{-1})$. Corresponding to each of the above forms, there exist unique self-adjoint operators denoted respectively by $H,  \, \wt{H}, \, \wh{H},$ and $\wt{\wt{H}}$. 
\end{theorem}
\begin{proof}
Among these four statements, we give the proof of the statement (iv).   The spirit of the reasoning for the other three cases is very similar. 
Since $\rho_\theta$ is a family of quantum states, and it is assumed that $\ker(\rho_\theta) = \{0\}$, all its eigenvalues $\lambda_j$'s are in the interval $0< \lambda_j<1$. Thus $\rho^{-\frac{1}{2}} >1$ and $\rho' + \|\rho'\| +1 \ge 1$ enabling the definition of the quadratic form $F_{\wt{\wt{H}}} [u,v]$ as 
\[ F_{\wt{\wt{H}}} [u,v] = \int_0^\infty \dd s \left \langle e^{-s \frac{\rho}{2}} u, (\rho' + \|\rho'\| +1)  e^{-s \frac{\rho}{2}} v \right \rangle \quad \text{ with } u,\, v \in D(\rho^{-\frac{1}{2}}). \]
It is easy to see that this is a positive-definite symmetric form giving rise to an inner product with respect to which $D(\rho^{-\frac{1}{2}})$ is a Hilbert space, since $F_{\wt{\wt{H}}} [u,u] \ge  \|\rho^{-\frac{1}{2}} u\|^2 > \|u\|^2$, and since $\rho^{-\frac{1}{2}}$ is a closed self-adjoint operator. Therefore $F_{\wt{\wt{H}}}$ is a closed symmetric form bounded below, and by Theorem 2.1 of \cite[pages 309 -- 323]{kato}, this defines a self-adjoint operator $\wt{\wt{K}}$ with 
\[ F_{\wt{\wt{H}}}[u,v] = \langle u, \wt{\wt{K}} v\rangle,\]
 for $u \in D( \rho^{-\frac{1}{2}})$ and $v \in D(\wt{\wt{K}}) \subseteq D(\rho^{-\frac{1}{2}})$. Finally, $\wt{\wt{H}}$ is set as $\wt{\wt{K}} -(\|\rho'\| +1)$ with \newline $D(\wt{\wt{H}}) = D(\wt{\wt{K}})$. 
\end{proof}

\begin{theorem}
The operators  $\rho^{\frac{1}{2}} \wh{H}(\theta) \rho^{\frac{1}{2}}$, 
$\rho^{\frac{1}{2}}H(\theta) \rho^{\frac{1}{2}}$, and $\rho^{\frac{1}{2}}\wt{\wt{H}}(\theta)\rho^{\frac{1}{2}}$ are bounded
and trace class operators.  
\end{theorem}
\begin{proof}
By assumption,  $\rho^{\frac{1}{2}} \wh{H}(\theta) \rho^{\frac{1}{2}} =
\rho_\theta' \in \mathcal{B}_1(\mathcal{H})$ proving the first part of the theorem.\\
In the BvN-LD, we have  
\[\rho^{\frac{1}{2}}{H}(\theta) \rho^{\frac{1}{2}} = \int_0^\infty
\dd u \frac{\rho^{\frac{1}{2}}}{\rho + u} \rho_\theta'
\frac{\rho^{\frac{1}{2}}}{\rho + u}.\]
Note that, by using polar decomposition  of $\rho' = |\rho'|^{\frac{1}{2}} V |\rho'|^{\frac{1}{2}}$ where $V$ is a partial isometry, we have 
\[\int_0^\infty
\dd u \left\|\frac{\rho^{\frac{1}{2}}}{\rho + u} \rho_\theta'
\frac{\rho^{\frac{1}{2}}}{\rho + u}\right\|_1  
\le \int_0^\infty
\dd u \left\| |\rho'|^{\frac{1}{2}} \left( \frac{\rho^{\frac{1}{2}}}{\rho + u}\right)  \right\|_2^2,\]
 and by equation (3.7) that 
\[\left\|  |\rho'|^{\frac{1}{2}} \frac{\rho^{\frac{1}{2}}}{\rho +u}
\right\|_2^2 = \tr \left[ |\rho'|^{\frac{1}{2}}  \frac{\rho}{(\rho
+u)^2} |\rho'|^{\frac{1}{2}} \right] = \tr \left[ |\rho'|  \frac{\rho}{(\rho
+u)^2} \right].\]
On the other hand,  $f - \int_0^M \dfrac{\rho}{(\rho + u)^2} f = \dfrac{\rho}{\rho+M} f$, which converges strongly to $0$ as $M \to \infty$ for any $f \in \mathcal{H}$, which combined with the above expression proves that 
\[\tr\left[ |\rho'| \int_0^\infty \frac{\rho}{(\rho + u)^2} \dd u \right] = \lim_{N \to \infty} \tr\left[ |\rho'| \int_0^N \frac{\rho}{(\rho + u)^2} \dd u \right] = \tr |\rho'| = \| \rho'\|_1.\]
Hence $\rho^{\frac{1}{2}} {H}(\theta)
\rho^{\frac{1}{2}}$ is a trace-class operator and therefore   
\begin{align*}
\tr \rho^{\frac{1}{2}} {H}(\theta) \rho^{\frac{1}{2}} &=
\int_0^\infty \dd u ~\tr \left(\frac{\rho^{\frac{1}{2}}}{\rho +u} \rho' \frac{\rho^{\frac{1}{2}}}{\rho +u} \right)\\
&= \int_0^\infty \dd u ~\tr \left(\frac{\rho}{(\rho+u)^2} \rho' \right)\\
&= \tr \left(\rho' \int_0^\infty \dd u ~\frac{\rho}{(\rho+u)^2} \right)= \tr\rho' =0.
\end{align*}

\par For $\wt{\wt{H}}(\theta)$, we  get by (4.2) 
$\rho^{\frac{1}{2}} \wt{\wt{H}} \rho^{\frac{1}{2}} = \int_0^\infty \dd s~~ \rho^{\frac{1}{2}} e^{-s \frac{\rho}{2}} \rho' e^{-s \frac{\rho}{2}} \rho^{\frac{1}{2}}$ and that 
\newline $\int_0^\infty \dd s ~~\| \rho^{\frac{1}{2}} e^{-s \frac{\rho}{2}} \rho' e^{-s \frac{\rho}{2}} \rho^{\frac{1}{2}}\|_1 \le \int_0^\infty \dd s~~\| |\rho'|^{\frac{1}{2}} \rho^{\frac{1}{2}} e^{-s \frac{\rho}{2}} \rho^{\frac{1}{2}}\|_2^2$. 
\par Consider a complete orthonormal system $\{f_j\}$ consisting of the eigenvectors of $|\rho'|^{\frac{1}{2}}$ corresponding to the eigenvalues $\mu_j$. Then $\rho^{\frac{1}{2}} e^{-s \frac{\rho}{2}} |\rho'|^{\frac{1}{2}} f_j = \mu_j \rho^{\frac{1}{2}} e^{-s \frac{\rho}{2}}$, and $\sum_j \mu_j^2 = \| \rho'\|_1 < \infty$. Thus 
\begin{align*}
\int_0^\infty \||\rho'|^{\frac{1}{2}} \rho^{\frac{1}{2}} e^{-s \frac{\rho}{2}}\|_2^2 ~\dd s 
&=\int_0^\infty \| \rho^{\frac{1}{2}} e^{-s \frac{\rho}{2}} |\rho'|^{\frac{1}{2}}\|_2^2 \dd s = \int_0^\infty \dd s ~\sum_j \|  \rho^{\frac{1}{2}} e^{-s \frac{\rho}{2}} |\rho'|^{\frac{1}{2}} f_j\|_2^2 \\
&= \sum_j \mu_j^2 \langle f_j, \int_0^\infty \rho e^{-s\rho} \dd s~ f_j\rangle 
 = \sum_j \mu_j^2 <\infty.
\end{align*}
The third equality uses the fact that the strong integral $\int_0^\infty \rho e^{-s\rho} \dd s = I$, and this proves the statement. Furthermore,
\[\tr (\rho^{\frac{1}{2}}\wt{\wt{H}}\rho^{\frac{1}{2}}) = \int_0^\infty \tr ( \rho^{\frac{1}{2}} \rho e^{-t\frac{\rho}{2}} \rho' \rho e^{-t\frac{\rho}{2}} \rho^{\frac{1}{2}} ) \dd t = \tr (\rho' \int_0^\infty \rho e^{-t\rho} \dd t)= \tr \rho' =0.\] 
\end{proof}
\begin{rem}
\begin{enumerate}[(i)]
\item In view of this result, it is natural that in infinite dimensional Hilbert space, for these LD's we redefine the expectation of the LD's in the state $\rho_\theta$ as $\mathbb{E}_\theta(\cdot) = \tr(\rho_\theta^{\frac{1}{2}} \cdot \rho_\theta^{\frac{1}{2}})$. As it is clear, if $\dim \mathcal{H}<\infty$, then much of these discussions are unnecessary and expectation (given above) reverts back to the more standard form. 
\item Furthermore, if $\dim(\mathcal{H}) = \infty$, the QFI in all the  four scenarios will not, in general, exist
        since the LD operators  are unbounded in each case.   However,  they may exist in  specific 
       cases ( see  Example 5.6).  
\end{enumerate} 
\end{rem}

\begin{eg}
Let $\mathcal{H}= \ell_2$ and $\{\ket{e_j}\}$ be a set of orthonormal
vectors. Consider the density matrix 
\[\rho_\theta =\sum_{j=0}^\infty \lambda_j(\theta) \ketbra{e_j}, \quad \text{ with } \lambda_j(\theta) =e^{-j \theta } (1 - e^{-\theta}), \; \, \theta > 0,\, j = 0,1,2,\cdots;  \]
$e_j$'s are independent of $\theta$ and $\sum_{j=0}^\infty \lambda_j(\theta)=1$. Thus 
\begin{eqnarray*}
\rho_\theta' &=& \sum_{j=0}^\infty \lambda_j'(\theta) \ketbra{e_j} \quad
\text{and} \quad \lambda_j'(\theta) = e^{-j\theta } [e^{-\theta} - j (1-
e^{-\theta})],
\end{eqnarray*}
so that  $\dfrac{\lambda_j'(\theta)}{\lambda_j(\theta)} = \dfrac{1}{e^\theta -1} -j $ which implies that 
$\left|\dfrac{\lambda_j'(\theta)}{\lambda_j(\theta)} \right| \to \infty $  as $j \to \infty$. However 
\newline $\sum_{j=1}^\infty |\lambda_j'(\theta)|< \infty$, implying that, though $\rho_\theta'$ is  $\mathcal{B}_1(\mathcal{H})$, $H_1(\theta) $ is an unbounded self-adjoint operator for each $\theta$ (see Theorem 3.6).  Here $H_2(\theta) =0$, since eigenvectors are independent of $\theta$. In fact  $\sum_{k=0}^\infty \dfrac{\lambda_k'(\theta)}{\lambda_k(\theta)} P_k$ is bounded if and only if $\sup_k \left|\dfrac{\lambda_k'(\theta)}{\lambda_k(\theta)}\right|$ is finite for each $\theta$. 

\end{eg}

\begin{eg}\label{pbox}
Consider the case of a particle in a box whose space
$\mathcal{H} = L^2[0,1]$ with a periodic boundary condition and 
 whose density matrix $\rho_\theta$ which is constructed in
the way, given below. Let the Laplacian on $\mathcal{H}$ be  given by 
\[D(\triangle_P) \equiv \{f \in L^2 :  \text{$f$ and $f'$ are both absolutely continuous, } f''\in L^2, \, f(1)=f(0)e^{\imath \theta}\}.\]
The eigenvalues are given by $\lambda_k (\theta) = (2 \pi k + \theta)^2$,
and eigenvectors are given by 
\newline $e_k(\theta) = \exp[\imath (2k\pi
+\theta )x]$, for all $k \in \mathbb{Z}$.  Construct the family of
density matrices as  
\begin{align*}
\rho_\theta &\equiv \frac{1}{c(\theta)}\sum_{k \in \mathbb{Z}} \frac{1}{(2k\pi +\theta )^2+1}  \ketbra{e_k(\theta)}, \quad \text{ where }\\
c(\theta) &= \sum_{k \in \mathbb{Z}} \frac{1}{(2k\pi +\theta )^2+1}, ~~\text{ is the normalising constant, and } \\
c'(\theta) &= \sum_{k \in \mathbb{Z}} \frac{-2(2k \pi +\theta)}{((2k\pi +\theta )^2+1)^2}.
\end{align*}
The eigenvalues of $\rho_\theta$ are 
\[\lambda_k(\theta) = \frac{c^{-1}(\theta)}{(2k\pi +\theta )^2+1}\quad k \in \mathbb{Z}.\]
Note that $|c'(\theta)| \le 2c(\theta)$, and  $\|e_k'(\theta)\|^2 = \int_0^1 x^2 \dd x=\frac{1}{3}$. $\sum_k \lambda_k(\theta) \|e_k'(\theta) \| = \frac{1}{\sqrt{3}} \sum_k \lambda_k(\theta) = \frac{1}{\sqrt{3}} < \infty$. Thus
\[\left| \frac{\lambda_k'(\theta)}{\lambda_k(\theta)}\right| = \left| - \frac{c'(\theta)}{c(\theta)} - \frac{2(2k\pi +\theta)}{(2k\pi +\theta)^2 	+1}\right| \le 3.\]
Therefore by 
$\sum_k|\lambda_k'(\theta)| = \sum_{k\in \mathbb{Z}}\left| \frac{\lambda_k'(\theta)}{\lambda_k(\theta)}\right| \lambda_k(\theta) < \infty$ and $\sum \lambda_k(\theta) \|e_k'(\theta) \|\le \frac{1}{\sqrt{3}} \sum\lambda_k(\theta)=\frac{1}{\sqrt{3}} $, which implies by Theorem \ref{new}  that $\rho_\theta' \in \mathcal{B}_1$. In this example $H_1(\theta) \in \mathcal{B(H)}$  by Theorem 3.6(i) and $H_2(\theta)$ is unbounded, self-adjoint operator with its domain containing all polynomials (note that the second hypothesis in the Theorem 3.6 is not satisfied ). 
\end{eg}

\begin{eg}[Coherent states in $\Gamma(\mathbb{C})$]
Let $\Gamma(\mathbb{C}) = \bigoplus \mathbb{C}^n$ be the bosonic Fock space with $\mathbb{C}^0 \simeq \mathbb{C}^1 = \mathbb{C} $ \cite{krp4}. It is convenient to have an unitary isomorphism of $\Gamma(\mathbb{C})$ with $L^2(\mathbb{R})$ by setting for every $z \in \mathbb{C}$, the exponential vector 
\begin{equation} \label{expv}
\bm{e}(z) = \sum_{n=0}^\infty \frac{z^n}{\sqrt{n!}} h_n \equiv \sum_{n=0}^\infty \frac{z^n}{\sqrt{n!}} \ket{n},
\end{equation}
where the family $\{h_n(x) \equiv C_n H_n(x) e^{-\frac{x^2}{2}}: \, n =0,1,2,\cdots\}$ $H_n$   the (monic) Hermite-polynomials,  is a complete orthogonal basis in $L^2(\mathbb{R})$ with $C_n$ suitable normalising constants (whose exact expression we shall not need here) and $\ket{n}$ is a notation for the same as $n$-particle state, more familiar to the physicists. In this notation, $\mathbb{C} \ni 1 \leftrightarrow \bm{e}(0) \equiv \ket{0}$, is called the vaccum. For other purposes, we introduce the normalised exponential vector 
\begin{align*}
\ket{z} &\equiv e^{-\frac{|z|^2}{2}} \bm{e}(z), \quad \text{ so that }\\
\langle z, w\rangle &= \exp \left( -\frac{1}{2}(|z|^2 + |w|^2) + \bar{z}w \right), \quad \text{ for } z, w \in \mathbb{C};
\end{align*}
leading to the \emph{continuous} resolution of the identity 
\begin{equation}\label{oc}
\int_\mathbb{C} \ketbra{z} \dd^2 z = I_{L^2(\mathbb{R})}.
\end{equation}
Furthermore, one has a canonically associated family of unitary Weyl operators defined as: for $z, w \in \mathbb{C}$,
\begin{align} \label{weyl1}
W(z) \bm{e}(w) &= \exp\left( -\frac{|z|^2}{2} - \bar{z}w \right) \bm{e} (z+w), \quad \text{ leading to } \\
W(z) \ket{w} &= \exp\left( - i ~\mathrm{Im} (\bar{z}w) \right) \ket{z+w}. \nonumber
\end{align}
In this background, we consider an expression 
\begin{equation}\label{weyl2}
\sigma = \int_\mathbb{C} \ketbra{z} \phi(z,\bar{z}) ~\dd^2 z, 
\end{equation}
where $ \phi(z,\bar{z}) $ is a probability density function and $ \dd^2 z$ is the Lebesgue measure in $\mathbb{C}$. Then we have: 
\begin{prop} \label{pro5.7}
$\sigma \in \mathbb{B}_{1,+} $ and $\tr \sigma =1$. 
\end{prop}
\begin{proof}
Let $\{f_m\}$ and $\{g_m\}$ be a family of bi-orthonormal set in $L^2(\mathbb{R})$ and we note that 
\[\sum_m \left| \langle f_m, \sigma g_m \rangle \right|  \le \int_\mathbb{C} \sum_m |\langle f_m, z \rangle | ~|\langle z, g_m \rangle| \phi(z, \bar{z}) ~\dd^2z.\]
This, on using Cauchy-Schwartz inequality, yields that 
\[  \sum_m \left| \langle f_m, \sigma g_m \rangle \right|  \le \int_\mathbb{C} \| \ket{z} \|^2 \phi(z, \bar{z}) \dd^2 z =1, \]
showing that $\sigma \in \mathbb{B}_1$. That $\sigma$ is non-negative is clear since $\phi$ is non-negative and 
\[\tr \sigma = \sum_n \langle h_n, \sigma h_n \rangle = \int_\mathbb{C} \sum_n \langle h_n, z \rangle \langle z,  h_n \rangle \phi(z, \bar{z}) \dd^2 z = \int_\mathbb{C} \| \ket{z} \|^2 \phi(z, \bar{z}) \dd^2 z =1. \]
\end{proof}
Next, we consider one real-parameter ($\theta$) of states 
\[ \rho_\theta \equiv \int_\mathbb{C} (\ketbra{z}) \phi_\theta (z,\bar{z}) ~\dd^2 z,\]
where the density function $\phi_\theta$ is a displaced Gaussian (parametrised by $M$), i.e. 
\[ \phi_\theta (z,\bar{z}) = \frac{1}{\pi M} \exp \left[ - \frac{(x-\theta)^2 +y^2}{M}\right], \]
where $z = x + i y$, with $M >0, ~ \theta \in \mathbb{R}$. Then we have 
\begin{theorem} \label{th5.8}
\begin{enumerate}[(i)]
\item If we set $\phi (z,\bar{z}) = \dfrac{1}{\pi M} \exp\left[ -\dfrac{1}{M} (x^2+y^2)\right]$ (with $z = x +i y$) and 
\newline $\rho_0 = \int_\mathbb{C} \ketbra{z} \phi (z,\bar{z}) ~\dd^2 z$, then $\rho_{0,+} \in \mathbb{B}_{1,+}$ with $\tr \rho_0 =1$. Furthermore, $\rho_\theta =W_0 (\theta) \rho_0 W_0 (\theta)^*$,  for all $ \theta \in \mathbb{R}$, where   we have written $W_0 (\theta) =W((\theta,0))$. 
\item $\mathbb{R} \ni \theta \mapsto \rho_\theta \in \mathbb{B}_1$ is real-analytic. 
\item $\rho_0$ admits the spectral representation $\rho_0 =\sum_{k=0}^\infty \lambda_k P_k \equiv \sum_{k=0}^\infty \lambda_k \ketbra{k}$, where 
\newline $\ket{k} \sim h_k \in L^2(\mathbb{R})$ are the $k$-particle states, and $\lambda_k = M^{-1} \left( 1 + \dfrac{1}{M} \right)^{-(1+k)}, \, (k =0,1,2, \cdots).$  
\item $\rho_\theta = W_0(\theta) \rho_0 W_0(\theta)^* = \sum_{k=0}^\infty \lambda_k P_k(\theta) $,  where  $P_k(\theta) = W_0(\theta) P_k W_0(\theta)^*$. 

\item $H_1(\theta)$, the first part of the LD operator in all scenarios is equal to $0$ identically, and $H_2(\theta)$(respectively $\wt{H}_2(\theta),\, \wh{H}_2(\theta),\, \wt{\wt{H}}_2(\theta)$) $= W_0(\theta) H_2(0)$ ($\wt{H}_2(0),\, \wh{H}_2(0),\, \wt{\wt{H}}_2(0)$) $W_0(\theta)^*  $. 

\item Let $\bm{I}_1(\theta)$ (respectively $\wt{\bm{I}}_1(\theta),\, \wh{\bm{I}}_1(\theta),\, \wt{\wt{\bm{I}}}_1(\theta)$) be component of the QFI involving derivatives of eigenvalues, and define $\bm{I}_2(\theta)\equiv \bm{I}(\theta)- \bm{I}_1(\theta)$ (respectively $\wt{\bm{I}}_2(\theta),\, \wh{\bm{I}}_2(\theta),\, \wt{\wt{\bm{I}}}_2(\theta)$). Then,  $\bm{I}_1(\theta)=0$ in all scenarios and $\bm{I}_2(\theta)$ (respectively $\wt{\bm{I}}_2(\theta),\, \wh{\bm{I}}_2(\theta),\, \wt{\wt{\bm{I}}}_2(\theta)$) $=\bm{I}_2(0)$ (respectively $\wt{\bm{I}}_2(0),\, \wh{\bm{I}}_2(0),\, \wt{\wt{\bm{I}}}_2(0)$).
\item In two scenarios,  the QFI's are  $\bm{I}(0)=\bm{I}_2(0)= 2 \ln \left( 1+ \dfrac{1}{M} \right)$, and 
\newline  $\wh{\bm{I}}(0) = \wh{\bm{I}}_2(0) = \dfrac{2+2M(1+M)(3+M)}{(1+M)^4}$. 
\end{enumerate}
\end{theorem}
\begin{proof}
\begin{enumerate}[(i)]
\item This follows from Proposition \ref{pro5.7}, since $\int_\mathbb{C} \phi(z,\bar{z}) \dd^2 z =1$. Furthermore, a simple computation shows that $W_0(\theta) \ket{z} = \ket{z + (\theta, 0)} \equiv \ket{ (x+\theta) +i y}$ and hence 
\begin{align*} 
W_0(\theta) \rho_0 W_0(\theta)^* &= \int_\mathbb{C} \ketbra{x+\theta + i y} \phi(x,y)~ \dd^2 z \\
&= \int_\mathbb{C} \ketbra{z} \frac{1}{\pi M} \exp \left[ -\frac{1}{M} \left\{ (x-\theta)^2 +y^2 \right\}\right] ~\dd^2 z = \rho_\theta .
\end{align*} 
\item For every $B \in \mathcal{B}(L^2(\mathbb{R}))$, 
\[ \tr (\rho_\theta B) = \int_\mathbb{C} \langle z, Bz\rangle \frac{1}{\pi M} \exp \left[ -\frac{1}{M} \left\{ (x-\theta)^2 +y^2 \right\}\right] ~\dd^2 z,\]
and an application of Lebesgue dominated convergence theorem yields analyticy, in fact, entirety, of the map $\mathbb{R} \ni \theta \mapsto \tr(\rho_\theta B)$. 
\item Applying $\rho_0$ on $n$-particle vector in $\mathcal{H}$, we get 
\begin{align*}
\rho_0 \ket{n} &= \frac{1}{\pi M}  \int_\mathbb{C} \ket{z} \langle z, n \rangle e^{-\frac{|z|^2}{M}}~\dd^2 z \\
&= \frac{1}{\pi M}  \int_\mathbb{C} e^{-|z|^2 \left( 1 + \frac{1}{M} \right) } \frac{\bar{z}^n}{\sqrt{n!}} \left( \sum_{m=0}^\infty \frac{z^m}{\sqrt{m!}} \ket{m} \right) ~\dd^2 z,
\end{align*}
which, on an application of Fubini's theorem justifying the interchange of integration summation and on observing furthermore that $\int_0^{2\pi} \bar{z}^n z^m ~\dd t = 2\pi r^{n+m} \delta_{nm}$ (where $z = r e^{i t}$), leads to $\rho_0 \ket{n} = \left( M^{-1} \left( 1 + \frac{1}{M} \right)^{-(n+1)} \right) \ket{n}$. Thus, the $n$-particle vectors are eigenvectors of $\rho_0$ with eigenvalues $\lambda_n$ and the completeness of the orthonormal set $\{ \ket{n}\}$ gives the spectral representation. 
\item This is an immediate consequence of the second part of (i), and so are the conclusions of (v) and (vi). 
\item[(vii)]   $\bm{I}_2(\theta) = \bm{I}_2(0) = \int_0^1 \tr \left( \rho_0 H_2(0) \rho_0^{1-t} H_2(0) \right) ~\dd t$, and $H_2(0)  =\sum_{k=0}^\infty \ln \lambda_k P_k'(0)$ (formally). First to compute 
\begin{align*}
P_n'(0) &\equiv \frac{\dd}{\dd\theta}  \left[ W_0(\theta) \ketbra{n} W_0(\theta)^*\right] \big|_{\theta=0}\\
&= \sqrt{n+1}\left(\ket{n+1}\bra{n} + \ket{n}\bra{n+1}\right) - \sqrt{n}\left(\ket{n}\bra{n-1} + \ket{n-1}\bra{n}\right) .
\end{align*}
Since  $\int_0^1 \dd t~\rho_\theta^t H(\theta) \rho_\theta^{1-t} H(\theta)  =  W(\theta ) [\int_0^1\dd t~ \rho_0^t H(0) \rho_0^{1-t} H(0)]  W(\theta)^* \equiv W(\theta ) J W(\theta )^*$
         the relevant   operator in the expression for $\bm{I}_2(\theta)$  is $=   W(\theta )   [\int_0^1 \dd t~ \rho_0^t H(0) \rho_0^{1-t} H(0)] W(\theta)^* = W(\theta) J W(\theta)^*$, which,   on integration,  reduces to 
         \begin{equation}
          =W(\theta) \left[ \sum_{\substack{k,l,m,n \\ k \neq l}} \left(\frac{\lambda_k-\lambda_l}{\ln \lambda_k- \ln\lambda_l}\right)  (\ln\lambda_m \ln \lambda_n)  (P_k P_m' P_l P_n') \right] W(\theta)^*
          \end{equation}

Next, we note from 
       the structure  of $P_l'(0)$ and Lemma  \ref{l1}(iii)  that  for every fixed $k =0,1,2,3,\cdots$,    the product    $P_k P_m'P_l P_n' = 0$  if $k=l$  and if $|k-m|,\, | l-m| ,\, | l-n|$ are  each  $\ge2$; and therefore it is a  rank-4 operator, with   estimate   $\| P_kP_m'P_lP_n' \| \le$  a polynomial in $k$, $Q(k)$.
       This reduces the (apparently) fourfold-infinite sum to a single infinite sum over  $\{k=0,1,2,3,... \} $. Furthermore,  it is easy to see that
       $  \frac{\lambda_k- \lambda_{k+1}}{ \ln \lambda_k - \ln \lambda_{k+1}}    =   \frac{\text{Constant}(M)} {(1+\frac{1}{M} )^{  (k+2)}},$  leading to  the estimate  that 
       \begin{equation}
 \|   \sum_{\substack{k,l,m,n \\ k \neq l}} \left(\frac{\lambda_k-\lambda_l}{\ln \lambda_k- \ln\lambda_l}\right)  (\ln\lambda_m \ln \lambda_n)  (P_k P_m' P_l P_n')   \|_1 \le \text{Constant}(M) \sum_{k=0}^\infty Q(k) \frac{1} {(1+\frac{1}{M} )^{  (k+2)}},
       \end{equation}
which converges. This shows that the operator in the expression of $\bm{I}_2$    is trace-class and therefore its trace is independent of $\theta$.  To calculate exactly the trace needs  the following simple evaluations given in Table 1. 
\begin{table}[!h]
 \begin{center}
\begin{tabular}{ |r|r||r|r| } 
 \hline
\rowcolor{gray!25} \rule{0pt}{15pt} Operator & Trace &Operator & Trace \\ \hline 
$P_k P_k' P_{k+1} P_k'$ & $k+1$ & 
$P_k P_{k+1}' P_{k+1} P_k'$ & $-(k+1)$ \\\hline
$P_k P_k' P_{k-1} P_k'$ & $k$ &
$ P_k P_{k-1}' P_{k-1} P_k'$ & $-k$ \\ \hline 
$P_k P_{k+1}' P_{k+1} P_{k+1}'$ & $k+1$  &
$P_k P_k' P_{k+1} P_{k+1}'$  & $-(k+1)$ \\ \hline 
$P_k P_{k-1}' P_{k-1} P_{k-1}'$ & $k$  &
$P_k P_k' P_{k-1} P_{k-1}'$ & $-k$ \\ \hline
\end{tabular}
\end{center}
\caption{List of terms with non-zero trace values.}
\end{table}
\pagebreak
\newline Putting these values in the above expression and using the fact that $\ket{0-1} = 0$, we get the final sum, after simplification, as 
\[
= \frac{1}{M^2} \ln\left(1+\frac{1}{M} \right) \left[ \frac{1}{\left(1+\frac{1}{M} \right)^2} + \sum_{k=1}^\infty \left\{ \frac{k+1}{\left(1+\frac{1}{M} \right)^{k+2}} + \frac{k}{\left(1+\frac{1}{M} \right)^{k+1}} \right\} \right]
= 2 \ln \left(1+\frac{1}{M} \right),\]
which is the expression of QFI in the BvN scenario.  

\par Similar considerations are also valid for the calculation of $\wh{\bm{I}}_2$  and we need to look at   $\wh{J} =   \rho'_0 \rho_0^{-1} \rho_0' = \sum_{ k,l,m} 
         \frac{\lambda_k \lambda_l }{    \lambda_m} P'_k P_m P'_l$ and note  that  $P'_k P_m P'_l = 0$, unless   for each fixed $k =0,1,2, \cdots,\; 
         | k-m| , \, | m- l |$  are each $\le1$, thereby reducing the triple-infinite sum to a single one over the index $k$ only.  Furthermore ,
       $\|   P'_k P_m P'_l \| \le $   a polynomial in $k$ and   for example ,  
       $\frac{\lambda_k \lambda_{k+1}} {\lambda_{k-1}}  =$  Constant$(M) \cdot
        ( 1+ M^{-1} )^{-(k+3)}$, thereby leading to the  conclusion that $\wh{J}$  is in $\mathcal{B}_1(\mathcal{H})$ and  a simple  computation   gives the value 
        \[ \wh{\bm{I}}_2(0) =  \frac{2+ M (2+M) (3+2M)}{(1+M)^4}. \]  
\end{enumerate}
\end{proof}
\end{eg}
\begin{rem}
\begin{enumerate}[(i)]
\item As can be seen from Theorem \ref{th5.8} ((iv), (v),  and (vi)), the QFIs in all scenarios are independent of the parameter $\theta$. This phenomenon, namely the independence of the QFI from the \emph{location} parameter $\theta$, is a routine property in the classical theory also, and this is due to the property of translation invariance. 
\item It is interesting to note   from the expressions  of QFI's in the 2 scenarios given in the previous theorem 
    that       $\lim_{M\to \infty} M \bm{I}(0,M)   = \lim_{M \to  \infty}  M  \wh{\bm{I}}(0,M)  = 2$.
This also implies that in the limit of $M\to \infty$, 
       the `Gaussian', intuitively, reduces to a  \emph{uniform}  distribution and in such a limit, both the QFI's  converge to 0. This
     can be interpreted as  the  \emph{complete  loss of  information}, as may  be  intuitively expected .
     \end{enumerate}
\end{rem}
   
   \section{Quantum two-level systems} \label{s6}
   \par In this section, we give explicit calculations for Fisher
information of a two-level system. Here  $\mathcal{H}= \mathbb{C}^2$. 
\subsection{Two level system with $\theta$ dependent eigenvalues and eigenvectors} We impose analyticity and set
\begin{eqnarray}
\rho_\theta &=& \lambda(\theta) \begin{bmatrix}
\cos^2 \theta & \sin\theta \cos \theta \\
\sin\theta \cos \theta & \sin^2\theta \end{bmatrix} +(1-\lambda(\theta))  \begin{bmatrix}
\sin^2 \theta & -\sin\theta \cos \theta \\
-\sin\theta \cos \theta & \cos^2\theta \end{bmatrix} \nonumber\\
&\equiv& \lambda(\theta) P_1(\theta) +(1- \lambda(\theta)) P_2(\theta), \label{2l1}
\end{eqnarray}
where $0<\lambda(\theta) <1$. Note that $P_1(\theta) + P_2(\theta)=I_2$,
and are mutually orthogonal. $I_2$ denotes the identity matrix in
this space. If we define 
\[P_1(0) =\begin{bmatrix}
1 & 0 \\
0 & 0 \end{bmatrix} \equiv P(0), \quad \text{ and } \quad P_2(0) =\begin{bmatrix}
0 & 0 \\
0 & 1 \end{bmatrix},\]
and if we set 
\[ U(\theta) = \begin{bmatrix}  \cos \theta & \sin\theta \\
-\sin \theta & \cos \theta \end{bmatrix},\]
(the matrix of rotations in two dimensions),  
then $U(\theta)^* P(0)U(\theta) = P(\theta)$ for all $\theta \in \mathbb{R}$ and $\theta \mapsto \lambda(\theta) $ can be any real valued analytic function taking values in $(0,1)$. 
 In this case $P_1(\theta) =P(\theta), \, P_2(\theta) =I_2-P(\theta) =P_1(\theta)^\perp,$ and 
$\rho_\theta = \lambda(\theta) P(\theta) + (1 - \lambda(\theta)) P(\theta)^\perp,$
for which 
\[\frac{\dd \rho_\theta}{\dd \theta} = \rho_\theta' = \lambda'(\theta)
\left(2P(\theta)-1\right) +(2\lambda(\theta) -1) P'(\theta),\]
with suitable domain of $\theta$ such that $\mathrm{ker} \rho_\theta
=\{0\}$. Since 	
\[\rho_\theta =U(\theta)^* \left[ \lambda(\theta)  P(0) + (1 -
\lambda(\theta)) P(0)^\perp \right] U(\theta),\]
where $U(\theta)$ are unitaries, this happens only if $\lambda(\theta)
\neq 0,\, 1$. Using the fact that
$P'(\theta) = U(\theta)^* P'(0)U(\theta),$ 
where $P'(0) = \begin{bmatrix} 0 & 1 \\1 & 0 \end{bmatrix}$, and
$P'(\theta)^2 = P'(0)^2 =I$.

   \noindent{\bf Fisher information for BvN scenario}
   In this case
   \begin{eqnarray*}
   {H}(\theta) = \int_0^\infty \frac{\dd u}{\rho +u} \rho'
\frac{1}{\rho +u}  &=& \frac{\dd}{\dd \theta} \ln \rho_\theta
   =  \frac{\lambda'(\theta)}{\lambda(\theta)} P- \frac{\lambda'(\theta)}{1-\lambda(\theta)}
P^\perp +\left( \ln \frac{\lambda(\theta)}{1- \lambda(\theta)} \right)P'.
   \end{eqnarray*}
 Using the fact  $P'(\theta)^2 = P'(0)^2=I$, one can show that  
 \[
 \bm{I}(\theta) = \mathbb{E}_\theta {H}(\theta)^2 \nonumber\\
 =\frac{\lambda'(\theta)^2}{\lambda(\theta)(1-\lambda(\theta))} +
\left(\ln  \frac{ \lambda(\theta)}{1- \lambda(\theta)} \right)^2
 = \bm{I}_1 (\theta) + \bm{I}_2 (\theta).
 \]
 \newline \noindent {\bf 2-level LD$_1$ scenario:} 
Fisher information in this model is given by 
\[\wt{\bm{I}}(\theta) = \frac{3}{4} \tr(\rho_\theta' \rho_\theta^{-1} \rho_\theta') + \frac{1}{4} \tr(\rho_\theta\rho_\theta' \rho_\theta^{-2} \rho_\theta').\]
Using Lemma \ref{l1} and the fact that there are only two orthonormal projections, namely $P(\theta)$ and $P(\theta)^\perp$, QFI takes the form 
 \begin{align*}
 \wt{\bm{I}}(\theta) &= \frac{3}{4} \left[ \frac{\lambda(\theta)'^2}{\lambda(\theta)(1
-\lambda(\theta))} + 2\frac{(2\lambda(\theta) -1)^2}{ \lambda(\theta)(1 -\lambda(\theta))}
\right] +\frac{1}{4} \left[ \frac{\lambda(\theta)'^2}{\lambda(\theta)
(1-\lambda(\theta)) } + \frac{2(2\lambda(\theta)-1)^2 (1-3\lambda(\theta)
+3\lambda(\theta)^2)}{\lambda(\theta)^2(1-\lambda(\theta))^2} \right] \nonumber\\
&= \frac{\lambda(\theta)'^2}{\lambda(\theta)(1 -\lambda(\theta))} + \frac{(2\lambda(\theta)
-1)^2}{2\lambda(\theta)^2 (1-\lambda(\theta))^2}
 = \wt{\bm{I}}_1 (\theta) + \wt{\bm{I}}_2 (\theta). 
 \end{align*}
  \noindent {\bf 2-level system in LD$_2$:}
 In this logarithmic derivative operator takes the form 
 \[
 \wh{H}(\theta) = \rho_\theta^{-\frac{1}{2}} \rho_\theta' \rho_\theta^{-\frac{1}{2}} 
 = \left[ \frac{\lambda(\theta)'}{\lambda(\theta)} P(\theta) - \frac{\lambda(\theta)'}{1- \lambda(\theta)}
P(\theta)^\perp  \right] + \frac{(2\lambda(\theta)-1)}{\sqrt{\lambda(\theta)(1-\lambda(\theta))}}  P(\theta)'.\]
Using the fact that $\tr P' P^2 = \tr P P' P=0$ and $\tr(P'^2)
=2 $, the Fisher information takes the form 
 \[\wh{\bm{I}}(\theta) = \tr \rho_\theta \wh{H}(\theta)^2 
 = \tr \left( \rho_\theta' \rho_\theta^{-1} \rho_\theta' \right) = \frac{\lambda(\theta)'^2}{\lambda(\theta)(1 -\lambda(\theta))} + \frac{2(2\lambda(\theta) -1)^2}{
\lambda(\theta)(1 -\lambda(\theta))}  
 = \wh{\bm{I}}_1(\theta) + \wh{\bm{I}}_2(\theta). \]

 \noindent {\bf LD$_3$ i.e. SLD scenario:}
 SLD is defined as the solution $\wt{\wt{H}}(\theta)$ of 
 $\frac{1}{2} (\wt{\wt{H}}(\theta) \rho_\theta + \rho_\theta \wt{\wt{H}}(\theta)) = \frac{\dd}{\dd \theta} \rho_\theta,$
 which is given by 
$ \wt{\wt{H}}(\theta) = \int_0^\infty e^{-t\frac{\rho_\theta}{2}}\rho'_\theta e^{-t\frac{\rho_\theta}{2}} \dd t$.
 Inserting the spectral decomposition of $\rho_\theta$ and the corresponding expression for $\rho_\theta'$, we get 
 \[ \wt{\wt{H}}(\theta) = \left(\frac{\lambda(\theta)'}{\lambda(\theta)} P(\theta) - \frac{\lambda(\theta)'}{1-\lambda(\theta)}  P(\theta)^\perp\right) +2(2\lambda(\theta) -1)P'(\theta) = \wt{\wt{H}}_1(\theta) + \wt{\wt{H}}_2(\theta).\]
 In this, the Fisher Information takes the form 
\[\wt{\wt{\bm{I}}}(\theta)= \tr \rho_\theta \wt{\wt{H}}(\theta)^2 = \frac{\lambda(\theta)'^2}{\lambda(\theta)(1-\lambda(\theta))} + 4(2\lambda(\theta)-1)^2 = \wt{\wt{\bm{I}}}_1(\theta) + \wt{\wt{\bm{I}}}_2(\theta).\]
The above expression is simplified by using formulae given in Lemma \ref{l1}.

 \par In each of these three models, we have written the
expression in two parts: viz. the first one contains terms with
$\lambda'$, which we may call this as  \emph{neo-classical} part,
whereas the second part comes from the derivative of the projection
operators, which we may call as  \emph{truly quantum} part. Furthermore the
first term, i.e. neo-classical term  in each of the expressions of QFI are the same,
i.e. 
 \[\bm{I}_1(\theta)=  \wt{\bm{I}}_1(\theta) =\wh{\bm{I}}_1(\theta) =\wt{\wt{\bm{I}}}_1(\theta)  =
\frac{\lambda(\theta)'^2}{\lambda(\theta)(1 -\lambda(\theta))}.\]
 The value of the second term, i.e. the quantum term, is given in the
Table \ref{t1} below. 
 \begin{table}[h]
 \begin{center}
\begin{tabular}{ |c|c|c|c| } 
 \hline
\rowcolor{gray!25} \rule{0pt}{20pt}$\bm{I}_2(\theta)$ & $\wt{\bm{I}}_2(\theta)$& $\wh{\bm{I}}_2(\theta)$ & $\wt{\wt{\bm{I}}}_2(\theta)$ \\ \hline 
\rule{0pt}{25pt}$\left[\ln  \left(\dfrac{
\lambda(\theta)}{1-\lambda(\theta)}\right) \right]^2 $  &
$\dfrac{(2\lambda(\theta) -1)^2}{2\lambda^2(\theta) (1-\lambda(\theta))^2}  $ &
$\dfrac{2(2\lambda(\theta) -1)^2}{ \lambda(\theta)(1 -\lambda(\theta))}$ &
$4(2\lambda(\theta) -1)^2$\\
 \hline
\end{tabular}
\end{center}
\caption{Comparison of the second term of QFI  in different models for two level systems. }\label{t1}
\end{table}
\par Furthermore, if $\rho_\theta$ commutes with $\rho_\theta'$ for
each $\theta$, then in all the above expressions, the second term i.e.
${I}_2(\theta),\,\wt{I}_2(\theta), \,\wh{I}_2(\theta)$, and $\wt{\wt{I}}_2(\theta)$ vanishes.
The same thing happens if $\lambda =1/2$ as well. Also, for any value of $\theta$, $\wt{\bm{I}}_2(\theta) \ge \wh{\bm{I}}_2(\theta)\ge \wt{\wt{\bm{I}}}_2(\theta)$.

\subsection{ 2-level system with $\theta$-independent eigenvalues and $\theta$-dependent eigenvectors} The next family of states is taken from an example of  Hayashi \cite{haya02}, where the $\theta$ we have used in the previous example was replaced by $\frac{\theta}{2}$. Hence, only brief calculations are given for QFI under different models. 
Consider the family of states  given by 
\begin{eqnarray}
 \rho_\theta &=& \frac{1+r}{2} \begin{bmatrix}
\cos^2\theta & \sin\theta \cos\theta \\
\sin\theta \cos\theta & \sin^2 \theta
\end{bmatrix} + \frac{1-r}{2} \begin{bmatrix}
\sin^2\theta & -\sin\theta \cos\theta \\
-\sin\theta \cos\theta & \cos^2 \theta
\end{bmatrix}\nonumber\\
&=&\frac{1+r}{2} P(\theta) +  \frac{1-r}{2} P(\theta)^\perp, \quad \text{ with } 0\le r <1. \label{2l2}
\end{eqnarray}
In equation \eqref{2l1}, replacing  the terms with $\lambda(\theta)$ by $\frac{1+r}{2}$,  one obtains the state in \eqref{2l2}. 
Furthermore, here the eigenvalues are $\theta$ independent, whereas the eigen-projections $P(\theta)$ and $P(\theta)^\perp$ are $\theta$ dependent. As a result, the (classical) part of LDs and hence QFIs, coming from derivatives of eigenvalues will all be zero, i.e. for the QFIs, ${\bm{I}}_1 = \wt{\bm{I}}_1=\wh{\bm{I}}_1  = \wt{\wt{\bm{I}}}_1=0$. Observe that in the second term, i.e. truly quantum term in the various models of QFI's, all are independent of the parameter $\theta$, and  can be shown in the tabular format as, 

 \begin{table}[h]
 \begin{center}
\begin{tabular}{ |c|c|c|c| } 
 \hline
\rowcolor{gray!25} \rule{0pt}{20pt}$\bm{I}_2(0)$ & $\wt{\bm{I}}_2(0)$& $\wh{\bm{I}}_2(0)$ & $\wt{\wt{\bm{I}}}_2(0)$ \\ \hline 
\rule{0pt}{25pt}$\left[\ln  \left(\dfrac{
1+r}{1-r}\right) \right]^2 $  &
$\dfrac{8r^2}{(1-r^2)^2}  $ &
$\dfrac{8r^2}{1-r^2}$ &
$4r^2$\\
 \hline
\end{tabular}
\end{center}
\caption{Comparison of the second term of QFI  in different models. }\label{t2}
\end{table}
Observe that $\bm{I}_2(0),\, \wt{\bm{I}}_2(0)$, and $\wh{\bm{I}}_2(0) \to \infty$ as $r \to 1$, whereas $\wt{\wt{I}}_2(0) \to 4$ as $r \to 1$. Furthermore, $\wt{\bm{I}}_2(0) \ge \wh{\bm{I}}_2(0) \ge \wt{\wt{\bm{I}}}_2(0)$.  At $r = 0$, i.e. $\rho_\theta =\frac{1}{2}I$, the state is totally degenerate, and all these quantities are equal to zero. 

\par Note that in this case,  the derivatives of the eigenvalues are equal to zero, as they are not $\theta$ dependent. But $\rho' \neq 0$ via the eigenvectors only. Therefore, the $H_1$'s are all equal to zero; subsequently, all the $\bm{I}_1$'s are equal to zero.  If eigenvalues are independent of $\theta$, and yet their contributions are not considered, then this will lead to an inconsistent conclusion because QFI due to such an error becomes 0. The examples 5.6 for coherent states and 6.2 for two-level systems emphasise this point.

\section*{Acknowledgement}
The first author is grateful to the Indian National Science Academy
for support through its Senior Scientist Programme and Prof.
Probal Chaudhuri for numerous fruitful discussions on classical
theory. 
\providecommand{\bysame}{\leavevmode\hbox to3em{\hrulefill}\thinspace}
\providecommand{\MR}{\relax\ifhmode\unskip\space\fi MR }
\providecommand{\MRhref}[2]{%
  \href{http://www.ams.org/mathscinet-getitem?mr=#1}{#2}
}
\providecommand{\href}[2]{#2}

\end{document}